\documentclass{lmcs}
\pdfoutput=1

\usepackage{lastpage}
\renewcommand{\dOi}{14(4:5)2018}
\lmcsheading{}{1--\pageref{LastPage}}{}{}%
{May~29,~2017}{Oct.~26,~2018}{}


\usepackage{hyperref}
\usepackage{subcaption}
\usepackage{microtype}

\usepackage{makecell}
\usepackage{multicol}
\usepackage{multirow}

\usepackage{tikz,tikz-qtree}
\usetikzlibrary{backgrounds,calc,positioning}

\definecolor{darkgreen}{rgb}{0.1,0.5,0.1}
\definecolor{darkred}{rgb}{0.6,0.1,0.0}

\tikzset{every tree node/.style={draw, circle, black, semithick, inner sep = 2pt, minimum size = 10pt}}
\tikzstyle{arrow} = [semithick,->,>=stealth]
\tikzset{edge from parent/.append style={arrow, edge from parent path = {(\tikzparentnode) -- (\tikzchildnode)}}}
\tikzset{level distance = 40pt, sibling distance = 20pt}
\tikzset{gate/.style={draw, circle, black, semithick, inner sep = 1.5pt}}
\tikzstyle{u} = [darkgreen!50]
\tikzstyle{etc} = [-,semithick,shorten <= .3cm,shorten >=.3cm,dash pattern=on \pgflinewidth off 6pt]

\newcommand{\C}{\mathcal{C}}

\newcommand{\Bisim}{\textsc{Bisim}}
\newcommand{\Sim}{\textsc{Sim}}

\newcommand{\N}{\mathbb{N}}
\newcommand{\DLOGTIME}{{\sf DLOGTIME}}
\renewcommand{\L}{{\sf L}}
\newcommand{\NL}{{\sf NL}}
\newcommand{\LogCFL}{{\sf LogCFL}}
\newcommand{\NC}{{\sf NC}}
\newcommand{\AC}{{\sf AC}}
\newcommand{\SAC}{{\sf SAC}}
\newcommand{\TC}{{\sf TC}}
\renewcommand{\P}{{\sf P}}
\renewcommand{\S}{\mathcal{S}}
\newcommand{\arrow}[1]{\stackrel{#1}{\longrightarrow}}

\newcommand{\auf}{\texttt{(}}
\newcommand{\zu}{\texttt{)}}
\newcommand{\Zu}{\texttt{]}}

\newcommand{\problem}[3]{\bigskip\noindent\begin{tabular}{@{}l l}\multicolumn{2}{l}{~#1}\\[3px]{\bf Given:} & #2 \\[2px] {\bf Question:} & #3\end{tabular}\medskip}

\DeclareMathOperator{\depth}{depth}

\begin{document}

\title{The Complexity of Bisimulation and Simulation on Finite Systems}

\author[M.~Ganardi]{Moses Ganardi}
\address{University of Siegen, Germany}
\email{ganardi@eti.uni-siegen.de}

\author[S.~G\"oller]{Stefan G\"oller}
\address{Laboratoire Specification et Verification (LSV), ENS de Cachan \& CNRS, France}
\email{goeller@lsv.ens-cachan.fr}

\author[M.~Lohrey]{Markus Lohrey}
\address{University of Siegen, Germany}
\email{lohrey@eti.uni-siegen.de}

\begin{abstract}
	
	In this paper the computational complexity of the (bi)simu\-lation problem over restricted graph classes is studied.
	For trees given as pointer structures or terms the (bi)simulation problem is complete for logarithmic space or $\NC^1$, respectively.
	This solves an open problem from Balc{\'a}zar, Gabarr\'o, and S\'antha.
	Furthermore, if only one of the input graphs is required to be a tree,
	the bisimulation (simulation) problem is contained in $\AC^1$ ($\LogCFL$).
	In contrast, it is also shown that the simulation problem is \P-complete already for graphs of bounded path-width.
	
\end{abstract}

\maketitle

\section{Introduction}

Courcelle's theorem states that every problem definable in monadic second-order logic (MSO) is solvable in linear time
on graphs of bounded tree-width.
In recent works by Elberfeld, Jakoby, and Tantau, techniques have been developed to transfer this famous result to low space and
circuit complexity classes \cite{ElberfeldJT10,ElberfeldJT12}.
In particular, the following logspace (resp., $\NC^1$) version of Courcelle's theorem was shown (see Section~\ref{sec-prel} for the necessary
definitions):
\begin{thmC}[\cite{ElberfeldJT10,ElberfeldJT12}] \label{thm-elberfeld}
	For a fixed MSO-sentence $\psi$ and a fixed constant $k$ one can check in logspace whether a given structure $\mathcal{A}$ of 
	tree-width at most $k$ satisfies $\psi$. If a tree decomposition of $\mathcal{A}$ of width $k$ is 
	given in term representation, then one can check in $\DLOGTIME$-uniform $\NC^1$ whether $\mathcal{A}$ satisfies $\psi$.
\end{thmC}
This result is a very powerful metatheorem, which can be applied to many computational problems.
On the other hand, there are important problems solvable in logspace on graphs of bounded tree-width that are not covered by
Theorem~\ref{thm-elberfeld}. One example is the graph isomorphism problem. Graph isomorphism is not MSO-definable even
over finite paths since two finite paths are isomorphic if they have the same length, but one cannot express in MSO that two finite
sets have the same size. Lindell \cite{Lindell92} has shown that isomorphism of trees is in logspace, and only very recently 
Elberfeld and Schweitzer \cite{ElberfeldS16} extended this result
to graphs of bounded tree-width.

In this paper, we are concerned with the complexity of simulation and bisimulation, which are of fundamental
importance in the theory of reactive systems, see e.g. \cite{AcIng07} for more background.
It is known that on finite state systems simulation and bisimulation are both \P-complete  \cite{BalcazarGS92},
and hence have no efficient parallel algorithm unless $\P = \NC$.
Surprisingly, no results on the complexity of (bi)simulation on natural subclasses of finite state systems 
are known (whereas there exists an extensive literature on (bi)simulation problems for various classes
of infinite state systems, like pushdown systems or Petri nets; see e.g.~\cite{Srba:roadmap:04}).
The authors of \cite{BalcazarGS92} pose this open question and suggest to consider the bisimulation problem on trees.
The above remark that tree isomorphism cannot be expressed in MSO applies to bisimulation on trees as well (two 
finite paths are bisimilar if and only if they are isomorphic). Moreover,  it is not clear whether there is a natural reduction 
of the bisimulation problem on trees to the logspace-solvable isomorphism problem for trees
(or even bounded tree-width graphs). 

In this paper, we determine the complexity of the bisimulation problem and simulation problem on several subclasses of finite state systems. 
More precisely, we show the following results; see also Table~\ref{table:complexity}.
\begin{enumerate}
	\item On trees the (bi)simulation problem is complete for logarithmic space (resp., $\NC^1$)
	if the trees are given as pointer structures (resp., in term representation).
	\item The bisimulation problem between a tree and a dag (or arbitrary graph) belongs to $\AC^1$ and is $\NL$-hard.
	\item The simulation problem is \P-hard (and hence \P-complete) already for graphs of bounded path-width.
	\item Simulation of a tree by a dag as well as simulation of a dag by a tree is \LogCFL-complete.
\end{enumerate}
Whether the bisimulation problem on graphs of bounded tree-width is in $\NC$ remains open.

We prove our results for the bisimulation problem for trees (statement (1) above) by a reduction to the evaluation problem 
for a new class of Boolean circuits that we call  {\em tree-shaped circuits}. These are circuits
that are composed in a tree-like fashion of smaller subcircuits.
The important parameter of such a circuit is its {\em width} which is defined to be
the maximal number of different paths from the root to an input gate in one of the above mentioned subcircuits.
We emphasize that the size of the subcircuits (number of gates) may be unbounded.
The main technical contributions of this paper are logspace- and $\NC^1$-evaluation algorithms
(depending on the input representation) for  tree-shaped circuits of bounded width.
These circuits should not be confused with circuits of bounded tree-width, which are known to have
logspace- and $\NC^1$-evaluation algorithms (depending on the representation) by the above Theorem~\ref{thm-elberfeld}.
We show how to partially unfold tree-shaped circuits of bounded width into circuits of bounded tree-width. This unfolding
is possible in $\TC^0$ (assuming the right representation of the circuit). Finally, the resulting bounded tree-width circuit
can be evaluated using Theorem~\ref{thm-elberfeld}. 
For the above logspace result we actually prove a stronger statement: A given tree-shaped circuit of size $N$ and width
$m$ can be evaluated in space $O(\log N \cdot \log m)$.

One should also mention the paper \cite{GrGRHeLa12}, where a logic $\mathsf{LREG}$ (which extends classical
first-order logic) is introduced. It is shown that $\mathsf{LREG}$ captures logspace on directed trees. Hence,
bisimulation for trees is expressible in $\mathsf{LREG}$. Due to the very technical definition of $\mathsf{LREG}$ we think
that it is not easier to express the bisimulation problem in $\mathsf{LREG}$ than giving a direct logspace algorithm.
Moreover expressibility in $\mathsf{LREG}$ does not imply that bisimulation for trees in term representation is in $\NC^1$.

A conference version of this paper appeared in \cite{GanardiGL16}.

{
\setlength\tabcolsep{.2cm}
\renewcommand{\arraystretch}{2}
\begin{table}[]
\centering
\begin{tabular}{r|c|c|c|cc|}
  & \thead{trees \\ (term rep.)} & \thead{trees \\ (pointer rep.)} & \thead{tree, \\ graph} & \multicolumn{1}{l|}{\thead{bounded \\ path-width}} & \thead{general \\ graphs} \\ \hline
bisimulation & \multirow{2}{*}{$\NC^1$-compl.} & \multirow{2}{*}{$\L$-compl.} & in $\AC^1$, $\NL$-hard & \multicolumn{1}{c|}{in $\P$, $\L$-hard} & \\ \cline{4-5}
simulation & & & $\LogCFL$-compl. & \multicolumn{2}{c|}{$\P$-compl.} \\ \hline
\end{tabular}
\caption{Parallel complexity of the (bi)simulation problem on restricted classes of finite graphs. 
%All complexity classes denote completeness results except from the bisimulation problem on graphs of bounded path-width.
}
\label{table:complexity}
\end{table}
}

\section{Preliminaries} \label{sec-prel}

\subsection{Graphs and trees}

A {\em (directed) graph} $G=(V,E)$ consists of a set of nodes $V$ and a set of edges $E(G) = E \subseteq V \times V$.
If $(u,v) \in E$, we call $v$ a {\em successor} of $u$.
A {\em path} (of length $n \ge 0$ from $v_0$ to $v_n$) in $G$ is a node sequence $v_0, v_1, \dots, v_n$ such that $(v_i,v_{i+1}) \in E$
for all $0 \le i < n$.
A graph is {\em acyclic} if there is no path of length $\ge 1$ from a node to itself.
We say that two nodes $u,v \in V$ are {\em connected} if there exists a path from $u$ to $v$ in the underlying
undirected graph $(V, E \cup \{ (v,u) \mid (u,v) \in E\})$.
A set of nodes $U \subseteq V$ is {\em connected} if any two nodes in $U$ are connected.

A graph $T$ is a {\em (rooted) tree} if there exists a node $r \in V(T)$, called the {\em root} of $T$,
such that for all $v \in V(T)$ there exists exactly one path from $r$ to $v$.
The {\em size} $|u|$ of a node $u$ in a tree is the size of the subtree rooted in $u$.
The {\em depth} of a node $u$, denoted by $\depth(u)$, is the length of the unique path from the root to $u$.
A node $u$ is an {\em ancestor} (resp., {\em proper ancestor}) of $v$, briefly $u \preceq v$ (resp., $u \prec v$),
if there exists a path (resp., a path of length $\ge 1$) in $T$ from $u$ to $v$.
Given a tree $T$ we say that a node $v$ is {\em between} nodes $u$ and $w$
if either $v \in \{u,w\}$, or $u$ and $w$ are connected in $T \setminus v$
(which is the subgraph of $T$ obtained by removing $v$).

A {\em node-labelled graph} $(V,E,\beta)$ is a graph $(V,E)$ together with a labelling function
$\beta \colon V \to A$ into a finite set $A$.
An {\em edge-labelled graph} $(V,E)$ consists of a set of nodes $V$ and a labelled edge relation
$E \subseteq V \times A \times V$.
We also write $u \arrow{a} v$ instead of $(u,a,v) \in E$.
Unlabelled graphs are also regarded as labelled graphs over a singleton label set.
An edge-labelled tree is an edge-labelled graph $(V,E)$ where the sets
$E_a = \{(u,v) \mid (u,a,v) \in E \}$ are pairwise disjoint for $a \in A$ and $(V,\bigcup_{a \in A} E_a)$ is a rooted tree.

We denote by Graphs, Unlabelled-Graphs, Trees and Unlabelled-Trees the classes of (unlabelled) graphs and (unlabelled) trees, respectively.

\subsection{Circuits} \label{sec-circuits}

A {\em (Boolean) circuit} $C = (G,\beta)$ is a node-labelled graph, where $G$ is acyclic and 
$\beta\colon V(G) \to \{x_1, \dots, x_n,0,1,\neg,\wedge,\vee\}$ for some $n$.
Nodes of $C$ are usually called {\em gates}. Gates labelled by $0$, $1$ (constant gates) or by a variable $x_i$ (input gates)
have no successors. Gates labelled by $\wedge$ or $\vee$ have at least one successor, and gates labelled by $\neg$ have
exactly one successor.
A {\em variable-free} circuit is a circuit without input gates. 
In a variable-free circuit every gate can be evaluated to either $0$ or $1$.
The question of deciding whether a given gate of a given variable-free circuit evaluates to $1$ is known
as the {\em circuit value problem}. It is one of the classical $\P$-complete problems, see \cite{GrHoRu95} 
for more details.

\iffalse
\subsection{Monadic second-order logic}
We briefly summarize the definition of monadic second-order (MSO) logic.
Fix a vocabulary $\tau$ and fix countable set of element variables $x_1, x_2, \ldots$ and set variables $X_1, X_2, \ldots$.
Atomic formulas have the form $x = y$, $R(x_1, \ldots, x_n)$ and $x \in X$ where $x, y, x_1, \ldots, x_n$ are element variables, $X$ is a set variable and $R \in \tau$ is an $n$-ary relation symbol.
If $\varphi$ and $\psi$ are formulas, then also $\neg \varphi$, $\varphi \vee \psi$, $\exists x \varphi$ and $\exists X \varphi$ are formulas.
\fi

\subsection{Computational complexity}

In this paper we will work with the following complexity classes:
\[
	\AC^0 \subseteq \TC^0 \subseteq \NC^1 \subseteq \L \subseteq \NL \subseteq \LogCFL \subseteq \AC^1 \subseteq \NC \subseteq \P.
\]
In the following we briefly define these classes. Of course, $\P$ denotes deterministic polynomial time. 
A function $f \colon \Sigma^* \to \Gamma^*$ is  {\em logspace-computable} if it can be computed on a deterministic Turing-machine with
a read-only input tape, a write-only output tape and a working tape whose length is bounded logarithmically in the input length; such a machine is also
called a logspace transducer.
We denote by $\L$ the class of languages which can be decided in logspace, i.e., for which the characteristic function is logspace-computable.
Throughout the paper we will use implicitly that compositions of logspace-computable functions are logspace-computable again.
The class $\NL$ is the set of languages that can be decided by nondeterministic Turing  machine in logarithmic space.
It is closed under complement by the famous theorem of Immerman-Szelepcs{\'{e}}nyi \cite{Immerman88,Szelepcsenyi88}.

The complexity class $\NC^i$ ($i \geq 0$) contains all languages $L \subseteq \{0,1\}^*$ such that there exists a circuit family $(C_n)_{n \in \N}$ where
\begin{enumerate}
\item $C_n$ has the input gates $\{x_1, \ldots, x_n, \neg x_1, \ldots, \neg x_n\}$, $\wedge$-gates and $\vee$-gates of fan-in two, and a distinguished output gate,
\item $C_n$ has size $n^{O(1)}$ and depth $O((\log n)^i)$,
\item and $C_n$ accepts $x \in \{0,1\}^n$ if and only if $x \in L$.
\end{enumerate}
If we allow in (1) $\wedge$-gates and $\vee$-gates of unbounded fan-in, we obtain the class $\AC^i$.
If we allow in (1) $\wedge$-gates, $\vee$-gates and majority-gates -- all of unbounded fan-in -- then we obtain the class $\TC^i$.
Finally, the class $\SAC^i$ is obtained by taking $\wedge$-gates of fan-in two and unbounded fan-in $\vee$-gates in (1).
One can easily show that $\NC^i \subseteq \SAC^i \subseteq \AC^i \subseteq \TC^i \subseteq \NC^{i+1}$ for all $i \geq 0$.
For the classes above (and inclusively) $\SAC^1$, we assume logspace uniformity, which means that there is a logspace
transducer that computes from the unary encoding of the number $n$ the $n$-th circuit $C_n$. For the classes
$\AC^0$, $\TC^0$ and $\NC^1$  logspace uniformity is too weak (the class $\NC^0$ will not be used). For these
classes one usually imposes the stronger  \DLOGTIME-uniformity condition on the circuit families.
For the  constant depth classes $\AC^0$ and $\TC^0$, $\DLOGTIME$-uniformity means that
for given binary coded gate numbers  $u, v$ of the $n$-th ciruit $C_n$, one can (i) compute the type of gate $u$ in time
$O(\log n)$ and (ii) check in  time $O(\log n)$ whether $u$ is an input gate for $v$.
Since $C_n$ has polynomially many gates, the gates of $C_n$ can be encoded by bit strings of length $O(\log n)$. Thus
the time bound $O(\log n)$ is linear in the input length $|u|+|v|$.
To define $\DLOGTIME$-uniformity for $\NC^1$-circuits one needs the so-called extended connection language. 
We do not have to go into details (which can be found in \cite{BIS90,Ru81,Vol99}), since we will not work with uniformity explicitly.
The union of all classes $\NC^i$ is denoted by $\NC$.

The definitions of the above circuit complexity classes 
can be easily extended to functions $I\colon \{0,1\}^* \to \{0,1\}^*$.
This can be done by encoding $I$ by the language $L_I = \{ 1^i 0 w \mid w \in \{0,1\}^*, \text{ the $i$-th bit of $f(w)$ is 1} \}$.
The main complexity class used in this paper is \DLOGTIME-uniform $\TC^0$. It can be seen as the extension of $\AC^0$ by the ability of counting.
Typical problems in $\TC^0$ are the computation of the sum, product and integer quotient of two binary encoded integers, and the sum and product of an arbitrary number of binary encoded integers \cite{HeAlBa02}.

The class $\LogCFL$ contains all languages $L \subseteq \{0,1\}^*$ which are logspace-reducible to a fixed context-free language. It coincides with logspace-uniform $\SAC^1$ \cite{Ven91}.
Another characterizations uses alternating logspace Turing machines with polynomial proof trees. Recall that an alternating Turing machine is a Turing machine 
where the set of states is partitioned into existential states and universal states. A configuration where the current state is existential (resp., universal) is called
a existential (resp., universal) configuration.
A existential (resp., universal) configuration is accepting
if and only if there exists an accepting successor configuration (resp., if all successor configurations are accepting).
In our definition of alternating machines
there is no explicit accepting state and to accept the machine has to branch to a universal configuration
without successor configurations.
It is well-known that $\P$ is the class
of languages that are accepted by alternating Turing machines in logarithmic space. To get the subclass $\LogCFL$ one has to require that there exists a polynomial $p(n)$
such that for every accepted input word $w$ of length $n$
there exists a proof tree of size at most $p(n)$. A proof tree of an alternating Turing machine for input $w$ is a finite tree where the nodes are labelled with
configurations such that the following conditions hold:
\begin{itemize}
\item The root is labelled with the initial configuration for input $w$.
\item If a node $v$ is labelled with an existential configuration $c$, then $v$ has exactly one child that is labelled with a successor configuration of $c$.
\item If a node $v$ is labelled with a universal configuration $c$, then $v$ has exactly one child $v'$ for each successor configuration $c'$ of $c$, and
$v'$ is labelled with $c'$.
\end{itemize}
Ruzzo \cite{Ru80} proved that $\LogCFL$ coincides with the class of languages that can be accepted by an alternating Turing machine in logarithmic space
and such that for every accepted input there exists a polynomial size proof tree.
It was shown by Borodin et al. that $\LogCFL$ is closed under complement \cite{BCDRT89}.

For more details on space and circuit complexity we refer to \cite{Vol99}.

\subsection{Tree representations}

The complexity of tree problems often depends on how the trees are represented.
Firstly, trees can be given as {\em pointer structures}
where the edge relation is given explicitly as a list of pairs consisting of two node names,
which is the standard encoding of graphs in general.
Secondly, trees can be given in {\em term representation} (or {\em bracket representation}):
The string $\auf\zu$ represents a tree of size 1.
If a tree $T$ has a root and direct subtrees $T_1, \dots, T_n$ which have term representations $r_1, \dots, r_n$,
then the string $\auf r_1\cdots r_n \zu$ is a term representation of $T$.
Notice that a tree can have multiple term representations because trees are unordered.
Thirdly, trees can be represented in {\em ancestor representation} where
we specify a list of all pairs $(u,v)$ where $u$ is an ancestor of $v$.
Elberfeld et al. showed that term and ancestor representations can be converted into each other in $\TC^0$ \cite{ElberfeldJT12}.
We will utilize this result
since operating on ancestor representations of trees is technically easier than on term representations.

To represent node-labelled trees (e.g., tree decompositions), in the pointer and ancestor representation
we append a list of pairs consisting of a node name and a node label.
In the term representation node labelled trees can be encoded by introducing for each label $a$
an opening bracket symbol $\auf_a$. For instance, $\auf_a \auf_b \zu \auf_a\zu \zu$ encodes a tree with an $a$-labelled root
and two children which are labelled by $b$ and $a$.
If the set of node labels is not fixed (this is the case for tree decompositions) we choose an arbitrary binary block code
for the opening brackets $\auf_a$.
To represent an edge-labelled tree we transform it into a node-labelled tree,
e.g. by assigning the label of an edge $(u,v)$ to its end point $v$ and labelling the root by a special symbol.
Let us finally mention that these coding details for labelled trees are only relevant for our $\NC^1$ lower bound 
in Section~\ref{bisi-lower-bound}, which refers to $\AC^0$-reductions. The reason is that different codings of labelled trees in term representation can
be transformed in $\TC^0$, but not necessarily in $\AC^0$, into each other.

\subsection{Tree-width and path-width}

A {\em tree decomposition} $(T,\beta)$ of a directed graph $G$ consists of a tree $T$ and a function $\beta\colon V(T) \to 2^{V(G)}$ which assigns to each node of $T$ a so called {\em bag} such that
\begin{itemize}
	\item for all $v \in V(G)$ the set $\{ t \in V(T) \mid v \in \beta(t) \}$ is non-empty and connected, and
	\item for all $(u,v) \in E(G)$ there exists $t \in V(T)$ such that $u,v \in \beta(t)$.
\end{itemize}
The {\em width} of $(T,\beta)$ is $\max_{t \in V(T)} |\beta(t)|-1$ and the {\em tree-width} of a graph $G$ is the
minimum width over all tree decompositions of $G$.
A tree decomposition $(T,\beta)$ is a {\em path decomposition} if $T$ is a path. The {\em path-width} of a
graph $G$ is the minimum width of all path decompositions of $G$.
Tree-width and path-width are also defined for node- and edge-labelled graphs via their underlying unlabelled graph.
More background on the notions of tree- and path-width can be found for instance in \cite{DowneyF13}.

\subsection{Bisimulation and simulation} \label{sec-bisi+sim}

\begin{defi}
A {\em bisimulation} on an edge-labelled graph $(V,E)$
is a binary relation $R \subseteq V \times V$ such that for all $(u,v) \in R$
the following conditions hold:
\begin{enumerate}
	\item \label{forth-condition} For all $u \arrow{a} u'$ there exists $v \arrow{a} v'$ such that $(u',v') \in R$.
	\item \label{back-condition} For all $v \arrow{a} v'$ there exists $u \arrow{a} u'$ such that $(u',v') \in R$.
\end{enumerate}
A relation $R$ that only satisfies condition \ref{forth-condition} for all $(u,v) \in R$ is called a {\em simulation}.
A (bi)simulation on two edge-labelled graphs is a (bi)simulation on their disjoint union.
Two nodes $u,v$ are called {\em bisimilar} if there exists a bisimulation $R$ such that $(u,v) \in R$.
We say that $u$ {\em is simulated by} $v$ if there exists a simulation $R$ such that $(u,v) \in R$.
\end{defi}

It is easy to see that the union of all bisimulations on a graph $G$ is a bisimulation again and an equivalence relation,
called the {\em bisimulation equivalence} $\sim_G$, or simply $\sim$ if $G$ is clear from the context.
Similarly,  the union of all simulations on a graph is a simulation again and a preorder,
called the {\em simulation preorder} $\sqsubseteq_{G}$, or simply $\sqsubseteq$ if $G$ is clear from the context. If $u$ belongs to a graph $G$ and $v$ belongs to 
a graph $H \neq G$, 
we write $u \sqsubseteq v$ (resp., $u \sim v$) if $u \sqsubseteq_{G \uplus H} v$ (resp., $u \sim_{G \uplus H} v$),
where $G \uplus H$ is the disjoint union of $G$ and $H$.
Note that in general $u \sim v$ is not equivalent to $u \sqsubseteq v \sqsubseteq u$.
For two graph classes $\C_1, \C_2$ we define the bisimulation problem and the simulation problem:

\problem{$\Bisim(\C_1,\C_2)$}{Two graphs $G_1 \in \C_1$, $G_2 \in \C_2$ and two nodes $v_1 \in V(G_1)$, $v_2 \in V(G_2)$}{Are $v_1$ and $v_2$ bisimilar?}

\problem{$\Sim(\C_1,\C_2)$}{Two graphs $G_1 \in \C_1$, $G_2 \in \C_2$ and two nodes $v_1 \in V(G_1)$, $v_2 \in V(G_2)$}{Is $v_1$ simulated by $v_2$?}

Note that $\Bisim(\C_1,\C_2)$ and $\Bisim(\C_2,\C_1)$ are equivalent, since bisimulation is symmetric, but this is not the case for simulation.
If $\C_1 = \C_2$ then we just write $\Bisim(\C_1)$ or $\Sim(\C_1)$.
It is known that on a finite edge-labelled graph $G = (V,E)$ bisimulation equivalence and simulation preorder can be computed as a fixed point as follows. First we set $u \sim_0 v$ for all $u, v \in V$. Inductively, we define $u \sim_{k+1} v$ if \begin{itemize}
\item for all $u \arrow{a} u'$ there exists $v \arrow{a} v'$ such that $u' \sim_k v'$, and
\item for all $v \arrow{a} v'$ there exists $u \arrow{a} u'$ such that $u' \sim_k v'$.
\end{itemize}
One easily sees that $\sim_i \, \supseteq \, \sim_{i+1}$ for all $i \in \N$. Hence there exists $n \in \N$ such that
$\sim_n \, = \, \sim_{n+i}$ for all $i \in \N$ and $\sim_n$ coincides with bisimulation equivalence $\sim$.
Using the approximants above the bisimulation equivalence can be decided on finite graphs in polynomial time by a partition refinement algorithm \cite{KanellakisS90}. 
The fastest known algorithm is due to Paige and Tarjan, which runs in time $O(m \log n)$ where $m$ is the number of edges and $n$ is the number of nodes \cite{PaigeT87}. 
The simulation problem on finite graphs is also decidable in polynomial time by a refinement algorithm \cite{CS01}.
Similar as above, one defines preorders $\sqsubseteq_i$ such that $\sqsubseteq_n \, = \, \sqsubseteq$ for some $n \in \N$.
It is known that both problems are $\P$-hard \cite{BalcazarGS92,SawaJ05,Srba01}, and therefore:

\begin{thmC}[\cite{BalcazarGS92,SawaJ05}] 
	$\Bisim${\rm (Graphs)} and $\Sim${\rm (Graphs)} are $\P$-complete.
	For both problems $\P$-hardness already holds for dags.
\end{thmC}

\section{Complexity of bisimulation} \label{sec-bisi}

\subsection{Bisimulation on trees} 

In this section we consider the bisimulation problem on edge-labelled trees.
Here we have two versions, depending on whether the trees are given as pointer structures or in term representation.
If the trees are given in term representation then the nodes can be specified by numbers which refer to an opening bracket
in the term. 
Clearly, we can restrict ourselves to the problem of deciding whether the roots of the two input trees are bisimilar,
since one can compute subtrees rooted in a given node in the pointer (term) representation in logspace ($\TC^0$).

\begin{thm} \label{thm-bisi-logspace}
	\Bisim{\rm (Trees)} is in $\L$ if the trees are given as pointer structures.
	%The bisimulation problem for trees given as pointer structures is in \L.
\end{thm}

\begin{thm}
	\label{thm:bisi-nc1}
	\Bisim{\rm (Trees)} is in $\NC^1$ if the trees are given in term representation.
	% The bisimulation problem for trees given in term representations is in $\NC^1$.
\end{thm}
We remark that the same complexity bounds hold for the tree isomorphism problem \cite{Buss97,JennerKMT03,Lindell92}.
In fact, converting a tree given as a pointer structure into term representation is $\mathsf{FL}$-complete \cite{BM95},
where $\mathsf{FL}$ is the class of functions computable in deterministic logspace.
In that sense, with the proper abstraction of the input representation the actual complexity of \Bisim{\rm (Trees)} is $\NC^1$.

\subsubsection{Tree-shaped Boolean circuits} \label{subsec-pointer}

%We start with showing a logspace upper bound for the bisimulation problem for trees in pointer
%representation. 

Bisimilarity between two edge-labelled trees $T_1$ and $T_2$ can be expressed as a Boolean circuit $C_\sim(T_1,T_2)$:
For all $u \in V(T_1), v \in V(T_2)$ such that $\depth(u) = \depth(v)$ the circuit contains a gate $x_{u,v}$,
which evaluates to true if and only if $u$ is bisimilar to $v$.
Since we want to test whether the roots of $T_1$ and $T_2$ are bisimilar,
it suffices to consider the reachable gates $x_{u,v}$ where $\depth(u) = \depth(v)$.
We define
\begin{equation} \label{bisi-formula}
x_{u,v} = \bigwedge_{u \arrow{a} u'} \bigvee_{v \arrow{a} v'} x_{u',v'} \wedge \bigwedge_{v \arrow{a} v'} \bigvee_{u \arrow{a} u'} x_{u',v'}.
\end{equation}
As usual we regard an empty conjunction (resp., disjunction) as $1$ (resp., 0).
In particular, if both $u$ and $v$ are leaves, then $x_{u,v} = 1$, and if exactly one of $u$ and $v$ is a leaf, then $x_{u,v}=0$.
An example circuit is shown in Figure~\ref{fig:tree-bisimulation-circuit}, where $T_1$ and $T_2$ are unlabelled.
Note that the circuit is composed in a tree-shaped form from smaller circuits. These smaller circuits correspond to the definition in \eqref{bisi-formula} and have 
the crucial property that there exist
exactly two paths from the root (of the smaller circuit) to an arbitrary leaf, which are highlighted in one subcircuit in Figure~\ref{fig:tree-bisimulation-circuit}.
This is the case because each gate variable $x_{u',v'}$ occurs once in the first conjunction and once in the second conjunction in \eqref{bisi-formula}.
In fact, we will show that circuits with such a path property can be evaluated in logspace,
and even in $\NC^1$ under a suitable representation.

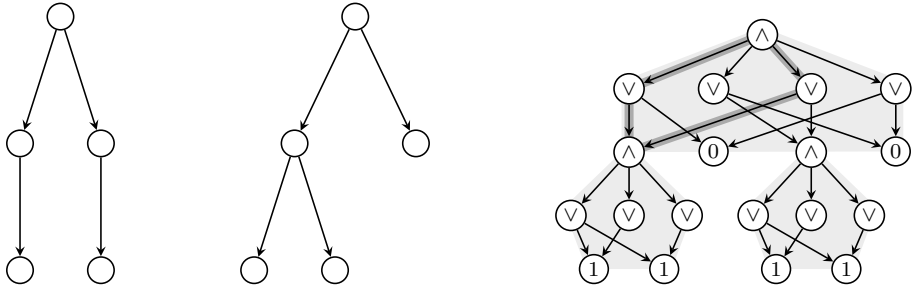
\begin{figure}[t]
	
	\centering
	
	\hspace*{\fill}
	\tikzset{level distance=48pt}
	\begin{tikzpicture}
	
	\Tree
	[. \node {};
	[. \node {};
	[. \node {}; ]
	]
	[. \node {};
	[. \node {}; ]
	]
	]
	
	\end{tikzpicture}
	\hspace*{\fill}
	\begin{tikzpicture}
	
	\Tree
	[. \node {};
	[. \node {};
	[. \node {}; ]
	[. \node {}; ]
	]
	[. \node {}; ]
	]
	
	\end{tikzpicture}
	\hspace*{\fill}
	\begin{tikzpicture}
	
	\tikzset{gate/.append style={fill=white}}
	\tikzstyle{xuv} = []
	\tikzstyle{bpath} = []
	
	\node[gate, xuv] (05) {\scriptsize $\wedge$};
	
	\node[gate, below left = 12pt and 10pt of 05] (35) {\scriptsize $\vee$};
	\node[gate, left = 20pt of 35] (15) {\scriptsize $\vee$};
	\node[gate, below right = 12pt and 10pt of 05] (06) {\scriptsize $\vee$};
	\node[gate, right = 20pt of 06] (09) {\scriptsize $\vee$};
	
	\draw[arrow] (05) -- (35);
	\draw[arrow, bpath] (05) -- (15);
	\draw[arrow, bpath] (05) -- (06);
	\draw[arrow] (05) -- (09);
	
	\node[gate, xuv, below = 12pt of 15] (16) {\scriptsize $\wedge$};
	\node[gate, xuv, below = 12pt of 35] (19) {\scriptsize $0$};
	\node[gate, xuv, below = 12pt of 06] (36) {\scriptsize $\wedge$};
	\node[gate, xuv, below = 12pt of 09] (39) {\scriptsize $0$};
	
	\draw[arrow, bpath] (15) -- (16);
	\draw[arrow] (15) -- (19);
	\draw[arrow] (35) -- (36);
	\draw[arrow] (35) -- (39);
	\draw[arrow, bpath] (06) -- (16);
	\draw[arrow] (06) -- (36);
	\draw[arrow] (09) -- (19);
	\draw[arrow] (09) -- (39);
	
	\node[gate, below = 12pt of 16] (17) {\scriptsize $\vee$};
	\node[gate, left = 10pt of 17] (26) {\scriptsize $\vee$};
	\node[gate, right = 10pt of 17] (18) {\scriptsize $\vee$};
	
	\draw[arrow] (16) -- (17);
	\draw[arrow] (16) -- (26);
	\draw[arrow] (16) -- (18);
	
	\node[gate, xuv, below left = 12pt and 5pt of 17] (27) {\scriptsize $1$};
	\node[gate, xuv, below right = 12pt and 5pt of 17] (28) {\scriptsize $1$};
	
	\draw[arrow] (17) -- (27);
	\draw[arrow] (18) -- (28);
	\draw[arrow] (26) -- (28);
	\draw[arrow] (26) -- (27);
	
	\node[gate, below = 12pt of 36] (37) {\scriptsize $\vee$};
	\node[gate, left = 10pt of 37] (46) {\scriptsize $\vee$};
	\node[gate, right = 10pt of 37] (38) {\scriptsize $\vee$};
	
	\draw[arrow] (36) -- (37);
	\draw[arrow] (36) -- (46);
	\draw[arrow] (36) -- (38);
	
	\node[gate, xuv, below left = 12pt and 5pt of 37] (47) {\scriptsize $1$};
	\node[gate, xuv, below right = 12pt and 5pt of 37] (48) {\scriptsize $1$};
	
	\draw[arrow] (37) -- (47);
	\draw[arrow] (38) -- (48);
	\draw[arrow] (46) -- (48);
	\draw[arrow] (46) -- (47);
	
	\begin{scope}[on background layer]
	\coordinate (A1) at (05);
	\coordinate (A2) at (15);
	\coordinate (A3) at (16);
	\coordinate (A4) at (39);
	\coordinate (A5) at (09);
	\coordinate (B1) at (16);
	\coordinate (B2) at (26);
	\coordinate (B3) at (27);
	\coordinate (B4) at (28);
	\coordinate (B5) at (18);
	\coordinate (C1) at (36);
	\coordinate (C2) at (46);
	\coordinate (C3) at (47);
	\coordinate (C4) at (48);
	\coordinate (C5) at (38);
	\fill[fill=gray,opacity=0.15]
	($(A1) + (0,.1)$) -- ($(A2) + (-.1,.1)$) -- ($(A3) + (-.1,0)$) -- ($(A4) + (.1,0)$) -- ($(A5) + (.1,.1)$);
	\fill[fill=gray,opacity=0.15]
	($(B1) + (0,.1)$) -- ($(B2) + (-.1,.1)$) -- ($(B3) + (-.1,0)$) -- ($(B4) + (.1,0)$) -- ($(B5) + (.1,.1)$);
	\fill[fill=gray,opacity=0.15]
	($(C1) + (0,.1)$) -- ($(C2) + (-.1,.1)$) -- ($(C3) + (-.1,0)$) -- ($(C4) + (.1,0)$) -- ($(C5) + (.1,.1)$);
	
	\coordinate (A6) at (06);
	\draw[gray, draw opacity=.6, line width=1.2mm] (A1) -- (A2) -- (A3);
	\draw[gray, draw opacity=.6, line width=1.2mm] (A1) -- (A6) -- (A3);
	\end{scope}
	
	\end{tikzpicture}
	\hspace*{\fill}
	
	\caption{Two trees $T_1, T_2$ and the tree-shaped circuit $C_\sim(T_1,T_2)$ for bisimulation equivalence.}
	\label{fig:tree-bisimulation-circuit}
	
\end{figure}

We use a more syntactic definition of such circuits:
A {\em tree-shaped circuit} is a sequence of Boolean equations $\S = (x_i = \varphi_i)_{1 \le i \le n}$
where $\varphi_1, \dots, \varphi_n$ are Boolean formulas over the variables $x_1, \dots, x_n$
such that the graph 
\begin{equation} \label{tree-T_S}
T_\S = (\{x_1, \dots, x_n\}, \{ (x_i,x_j) \mid x_j \text{ occurs in } \varphi_i \})
\end{equation}
is a tree with root $x_1$.
This implies that there are no cyclic definitions in $\S$ and that no variable $x_k$ appears in two distinct
formulas $\varphi_i$ and $\varphi_j$ ($i \neq j$). 
The {\em size} of $\S$ is defined as the sum of the sizes of all formulas $\varphi_i$.
The {\em width} of $\S$ is defined as the maximal number of occurrences of a variable $x_j$ in a formula $\varphi_i$.
An example of a tree-shaped circuit of width two is given by the formulas
\eqref{bisi-formula} for bisimulation.
We can view $\S$ as an ordinary Boolean circuit by taking the disjoint union of the formula trees of the $\varphi_i$
and then merging all $x_i$-labelled leaves with the root of the formula tree of $\varphi_i$.
For example, the tree-shaped circuit $C_\sim(T_1,T_2)$ for bisimulation equivalence can be regarded as a tree-shaped circuit of width 2,
which can be computed in logspace from $T_1$ and $T_2$. The main goal of this section is to show
that the circuit value problem restricted to tree-shaped circuits of bounded width belongs to logspace.

\begin{figure}
\begin{tikzpicture}[semithick, scale=0.8, every node/.style={scale=0.8, circle, draw, inner sep = 0, minimum size=.5cm}]

\node (x1) {\scriptsize $x_1$};
\node[below = 1cm of x1] (x3) {\scriptsize $x_3$};
\node[left = 2cm of x3] (x2) {\scriptsize $x_2$};
\node[right = 2cm of x3] (x4) {\scriptsize $x_4$};
\node[below left = 1cm and .3cm of x2] (x5) {\scriptsize $x_5$};
\node[below right = 1cm and .3cm of x2] (x6) {\scriptsize $x_6$};
\node[below left = 1cm and .3cm of x3] (x7) {\scriptsize $x_7$};
\node[below right = 1cm and .3cm of x3] (x8) {\scriptsize $x_8$};
\node[below left = 1cm and .3cm of x4] (x9) {\scriptsize $x_9$};
\node[below right = 1cm and .3cm of x4] (x10) {\scalebox{.8}{\scriptsize $x_{10}$}};

\draw (x2) -- ++(-1cm,0) -- (x1);
\draw (x2) -- (x3);
\draw (x3) -- (x4);
\draw (x4) -- ++(1cm,0) -- (x1);
\draw (x5) -- ++(-.5cm,0) -- (x2);
\draw (x6) -- ++(.5cm,0) -- (x2);
\draw (x5) -- (x6);
\draw (x7) -- ++(-.5cm,0) -- (x3);
\draw (x8) -- ++(.5cm,0) -- (x3);
\draw (x7) -- (x8);
\draw (x9) -- ++(-.5cm,0) -- (x4);
\draw (x10) -- ++(.5cm,0) -- (x4);
\draw (x9) -- (x10);
\draw (x5) -- ++(-.4cm,-.5cm);
\draw (x5) -- ++(.4cm,-.5cm);
\draw (x6) -- ++(-.4cm,-.5cm);
\draw (x6) -- ++(.4cm,-.5cm);
\draw (x7) -- ++(-.4cm,-.5cm);
\draw (x7) -- ++(.4cm,-.5cm);
\draw (x8) -- ++(-.4cm,-.5cm);
\draw (x8) -- ++(.4cm,-.5cm);
\draw (x9) -- ++(-.4cm,-.5cm);
\draw (x9) -- ++(.4cm,-.5cm);
\draw (x10) -- ++(-.4cm,-.5cm);
\draw (x10) -- ++(.4cm,-.5cm);

\draw[u] (x1) -- (x2);
\draw[u] (x1) -- (x3);
\draw[u] (x1) -- (x4);
\draw[u] (x2) -- (x5);
\draw[u] (x2) -- (x6);
\draw[u] (x3) -- (x7);
\draw[u] (x3) -- (x8);
\draw[u] (x4) -- (x9);
\draw[u] (x4) -- (x10);

\draw[darkred] (x3) edge [bend left = 25] (x7);
\draw[darkred] (x3) edge [bend left = 5] (x7);
\draw[darkred] (x3) edge [bend right = 10] (x7);
\end{tikzpicture}

\caption{A tree-shaped circuit $\S$ where the underlying tree $T_\S$ is drawn in green.}
\label{fig:tree-shaped-circuit}
	
\end{figure}
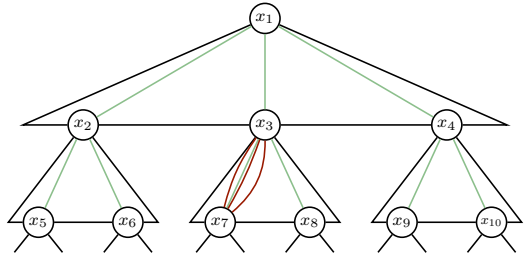

Note that a variable-free Boolean circuit can be represented by a relational structure with a binary
edge relation and unary relations for the labels $0$, $1$, $\neg$, $\wedge$, $\vee$.
Moreover, there is a formula of monadic-second order logic (MSO) expressing that a 
variable-free Boolean circuit evaluates to $1$. As a consequence, 
by Theorem~\ref{thm-elberfeld}, we can evaluate variable-free Boolean 
circuits of tree-width at most $k$ in logspace for every fixed $k$. 
Moreover, the complexity can be improved to $\NC^1$ if we also provide
for the input circuit a bounded width tree decomposition in term representation. 
For this, one stores a tree decomposition $(T,\beta)$ by an expression for the node-labelled
tree $T$, where every node $t \in V(T)$ is labelled by the bag $\beta(t)$.
Note that $\beta(t)$ is a set of gates of the circuit,
and these gates are stored by their addresses.
To sum up, we have:

\begin{thmC}[\cite{ElberfeldJT10,ElberfeldJT12}] \label{theo-coro-elberfeld}
	For every fixed $k \in \N$, the circuit value problem restricted to circuits of tree-width at most $k$ 
	can be solved in logspace. If in addition to the input circuit $C$ a width-$k$ tree decomposition of $C$ in term representation is given, then
	the circuit value problem can be solved in $\NC^1$.
\end{thmC}

\subsubsection{(Un)related restricted circuit classes}
We cannot directly apply Theorem~\ref{theo-coro-elberfeld} to tree-shaped circuits
since neither their tree-width nor clique-width is bounded by the following lemma. Clique-width
is another graph measure defined by so called $k$-expressions, see e.g. \cite{KaminskiLM09} for a survey. 
It is easy to see that Theorem~\ref{theo-coro-elberfeld}
can be generalized to circuits of clique-width $k$, provided a $k$-expression for the circuit is part of the input.

\begin{lem}
	For $n \ge 1$, let $T_n$ be the tree of size $n+1$ whose root has exactly $n$ children.
	The set of circuits $\{ C_\sim(T_n,T_n) \mid n \ge 1 \}$ has unbounded tree-width and unbounded clique-width.
\end{lem}

\begin{proof}	
	The underlying undirected graph of $C_\sim(T_n,T_n)$ contains the complete bipartite graph $K_{n,n}$
	as a (topological) minor. This can be seen by removing the output gate and dissolving all constant gates
	({\em dissolving} a node of degree 2 means deleting it and connecting its two neighbors).	
	Figure \ref{fig:unbounded-width} shows an example for $n=3$. Since the tree-width of $K_{n,n}$ is $n$ 
	and the tree-width of every minor of a graph $G$ is bounded by the tree-width of $G$, it follows that 
	the tree-width of $C_\sim(T_n,T_n)$ is at least $n$.
	
	It is known that a set of graphs which has bounded clique-width and for which there exist only finitely many $n$ such that the
	bipartite graph $K_{n,n}$ is contained as a subgraph (not only minor) also has bounded tree-width \cite{GurskiW00}.
	We claim that for each $n \ge 1$ the undirected graph of $C_\sim(T_n,T_n)$ does not contain $K_{3,3}$ as a subgraph,
	which implies that $\{ C_\sim(T_n,T_n) \mid n \ge 1 \}$ also has unbounded clique-width (note that
	$C_\sim(T_n,T_n)$ contains a $K_{2,2}$, i.e., a cycle on four nodes). Note that in $C_\sim(T_n,T_n)$ for every simple path of four nodes,
	one of the four nodes has degree $2$ (these are the gates in the middle layer of Figure \ref{fig:unbounded-width}). 
	But this is not possible in a $K_{3,3}$.
\end{proof}

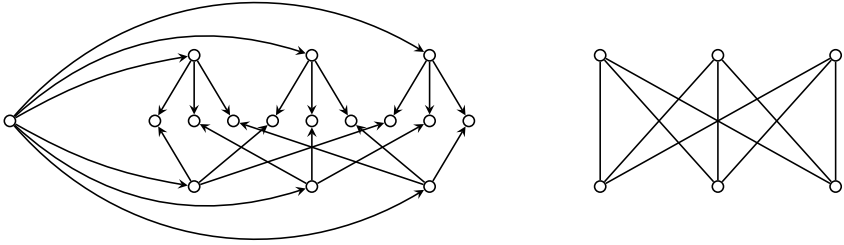
\begin{figure}[t]
	
	\centering
	
	\begin{tikzpicture}
	
	\node[gate] (11) {};
	\node[gate, right = 10pt of 11] (12) {};
	\node[gate, right = 10pt of 12] (13) {};
	\node[gate, right = 10pt of 13] (21) {};
	\node[gate, right = 10pt of 21] (22) {};
	\node[gate, right = 10pt of 22] (23) {};
	\node[gate, right = 10pt of 23] (31) {};
	\node[gate, right = 10pt of 31] (32) {};
	\node[gate, right = 10pt of 32] (33) {};
	
	\node[gate, above = 20pt of 12] (1-) {};
	\node[gate, above = 20pt of 22] (2-) {};
	\node[gate, above = 20pt of 32] (3-) {};
	
	\node[gate, below = 20pt of 12] (-1) {};
	\node[gate, below = 20pt of 22] (-2) {};
	\node[gate, below = 20pt of 32] (-3) {};
	
	\node[gate, left = 50pt of 11] (o) {};

	\draw[arrow] (1-) -- (11);
	\draw[arrow] (1-) -- (12);
	\draw[arrow] (1-) -- (13);
	\draw[arrow] (2-) -- (21);
	\draw[arrow] (2-) -- (22);
	\draw[arrow] (2-) -- (23);
	\draw[arrow] (3-) -- (31);
	\draw[arrow] (3-) -- (32);
	\draw[arrow] (3-) -- (33);
	
	\draw[arrow] (-1) -- (11);
	\draw[arrow] (-2) -- (12);
	\draw[arrow] (-3) -- (13);
	\draw[arrow] (-1) -- (21);
	\draw[arrow] (-2) -- (22);
	\draw[arrow] (-3) -- (23);
	\draw[arrow] (-1) -- (31);
	\draw[arrow] (-2) -- (32);
	\draw[arrow] (-3) -- (33);
	
	\draw[arrow] (o) edge [bend right = 10] (-1);
	\draw[arrow] (o) edge [bend right = 30] (-2);
	\draw[arrow] (o) edge [bend right = 40] (-3);
	
	\draw[arrow] (o) edge [bend left = 10] (1-);
	\draw[arrow] (o) edge [bend left = 30] (2-);
	\draw[arrow] (o) edge [bend left = 40] (3-);

	\node[gate, right = 30pt of 33, draw=none] (11') {};
	\node[gate, right = 10pt of 11', draw=none] (12') {};
	\node[gate, right = 10pt of 12', draw=none] (13') {};
	\node[gate, right = 10pt of 13', draw=none] (21') {};
	\node[gate, right = 10pt of 21', draw=none] (22') {};
	\node[gate, right = 10pt of 22', draw=none] (23') {};
	\node[gate, right = 10pt of 23', draw=none] (31') {};
	\node[gate, right = 10pt of 31', draw=none] (32') {};
	\node[gate, right = 10pt of 32', draw=none] (33') {};
	
	\node[gate, above = 20pt of 12'] (1-') {};
	\node[gate, above = 20pt of 22'] (2-') {};
	\node[gate, above = 20pt of 32'] (3-') {};
	
	\node[gate, below = 20pt of 12'] (-1') {};
	\node[gate, below = 20pt of 22'] (-2') {};
	\node[gate, below = 20pt of 32'] (-3') {};
	
	\draw[arrow,-] (1-') -- (-1');
	\draw[arrow,-] (1-') -- (-2');
	\draw[arrow,-] (1-') -- (-3');
	\draw[arrow,-] (2-') -- (-1');
	\draw[arrow,-] (2-') -- (-2');
	\draw[arrow,-] (2-') -- (-3');
	\draw[arrow,-] (3-') -- (-1');
	\draw[arrow,-] (3-') -- (-2');
	\draw[arrow,-] (3-') -- (-3');
	\end{tikzpicture}
	\caption{The circuit $C_\sim(T_3,T_3)$ for bisimulation equivalence.
		Removing the output gate on the left and dissolving the constant gates in the middle layer shows that $K_{3,3}$ is
		a minor of $C_\sim(T_3,T_3)$.}
	\label{fig:unbounded-width}
	
\end{figure}

Let us also compare tree-shaped circuits of constant width with a class of circuits that is presented in
\cite[page 5]{GrGRHeLa12}. For the latter, the authors require that for every path in the circuit
the product of the fan-outs (i.e., indegrees, since we direct circuits towards the input gates) of the gates on the path is bounded polynomially by the circuit size.
It is shown that circuits with this property can be evaluated in logspace. Note that 
tree-shaped circuits of constant width do not have this path property from 
\cite{GrGRHeLa12}. On the other hand, the circuits from  \cite{GrGRHeLa12} 
can have nodes with large fan-out, which is not possible for  tree-shaped circuits of constant width.
Hence, the two circuit classes are incomparable.

\subsubsection{Evaluating tree-shaped circuits given in pointer representation}

Before we show how to evaluate tree-shaped circuits of bounded width in logspace,
let us briefly recall the standard logspace algorithm which performs a depth-first traversal of a tree $T$.
It stores (i) a pointer to the current node of $T$ using $O(\log |T|)$ bits and 
(ii) a single bit, indicating the current direction of the traversal (up or down).
We will call this procedure {\em logspace depth-first traversal}.
If $T$ is a formula then one can evaluate the formula using the logspace algorithm above
by additionally storing the value of the current subformula.

Next we introduce some notions. Recall that the size of the subtree rooted in a node $u$ is denoted by $|u|$.
We say that a node $u$ is {\em heavy} if for all siblings $v$ of $u$ we either have
(i) $|u| > |v|$ or (ii) $|u|=|v|$ and $u < v$ (where $<$ denotes some fixed order on the nodes of the tree).
Otherwise a node is called {\em light}.
Notice that the root is heavy and that every inner node has exactly one heavy child.
Using the logspace depth-first traversal algorithm we can compute in logspace
for each node its size and determine whether it is heavy.
Note that every path in a tree contains at most $O(\log n)$ light nodes.

By the following result, tree-shaped circuits of bounded width can be evaluated in logspace (take $m = O(1)$).
Also note that a tree-shaped circuit of width $m=1$ is a tree and therefore can be evaluated in logspace.

\begin{thm}
	\label{thm:eval-logspace}
	A given tree-shaped circuit $\S = (x_i = \varphi_i)_{1 \le i \le n}$ of width $m \geq 2$ 
	can be evaluated in space $O(\log s + \log n \cdot \log m)$, where
	$s = \max \{|\varphi_i| \mid 1 \leq i \leq n \}$ is the maximal size of one of the formulas $\varphi_i$. In particular, 
	for every fixed $m \in \N$, there is a logspace algorithm which evaluates a given tree-shaped circuit of width at most $m$.
\end{thm}

\begin{proof}
	Let $\S = (x_i = \varphi_i)_{1 \le i \le n}$ be a tree-shaped circuit of width $m$ and let
	$s = \max \{|\varphi_i| \mid 1 \leq i \leq n \}$.
	Recall the definition of the tree $T_\S$ with node set $\{x_1, \ldots, x_n\}$ from \eqref{tree-T_S}.
	
	First of all, we label every node $x_i$ of $T_\S$ by (i) 
	the size $|x_i|$ of the subtree rooted in $x_i$ and (ii) a single bit indicating whether $x_i$ is a heavy node of $T_\S$. 
	This information can be computed in space $O(\log n)$ using a logspace depth-first traversal. 
	
	The circuit $\S$ is evaluated in a recursive way using a pointer to one of the nodes $x_1, \ldots, x_n$ (which needs space $O(\log n)$)
	and a stack of height $O(\log n \cdot \log m)$ as follows. Initially the pointer is set to the root $x_1$.
	Assume that $x_k$ is the heavy child of $x_1$ in $T_\S$.  Then, the pointer is moved to $x_k$
	without writing anything on the stack.
	Next, the subcircuit rooted at $x_k$ is evaluated recursively. 
	By induction, space $O(\log s + \log n \cdot \log m)$ is used for this.
	Once the algorithm returns from the recursion
	the pointer is back on $x_k$ and one can release the space. The value 
	of $x_k$ is stored on the stack, which now contains a single bit.
	Note that the algorithm knows (using the labeling computed before) that $x_k$
	is the heavy child of its parent node $x_1$. This information now triggers the evaluation
	of the Boolean formula $\varphi_1$. This is done by a logspace depth-first traversal
	of the Boolean formula tree of $\varphi_1$.
	Each time, this depth-first traversal of $\varphi_1$ arrives at a leaf node, the following is done:
	\begin{itemize}
		\item If the leaf is labelled with the variable $x_k$ (the heavy child of $x_1$ in $T_\S$), then the value of $x_k$ is retrieved from the stack and 
		the algorithm continues the evaluation of $\varphi_1$.
		\item If the leaf is a light child labelled with the variable $x_i \neq x_k$, then let $1 \leq c \leq m$ such that 
		the leaf corresponds to the $c$-th occurrence of $x_i$ in $\varphi_1$. The algorithm stores $c$ on the stack (which 
		needs space $O(\log m)$) and continues recursively with the evaluation of the subcircuit
		rooted at $x_i$. Once it comes back from the recursion, the pointer is back on $x_i$. The algorithm sees that $x_i$ is a light child
		of $x_1$. Using this information and the number $c$ stored on the stack, it continues the evaluation of $\varphi_1$ at
		the right position in $\varphi_1$.
	\end{itemize}
	Note that in the second case, we have $|x_i| \leq |x_1|/2$. This implies that the number of bits stored on the stack is bounded
	by $O(\log n \cdot \log m)$. The total space consumption is therefore $O(\log n + \log s + \log n \cdot \log m) = O(\log s + \log n \cdot \log m)$
	(since we assume $m \geq 2$).
\end{proof} 

Note that we do not assume $m$ to be a constant in Theorem~\ref{thm:eval-logspace}. In particular, since $k$, $n$, and $m$ are all bounded
by the size of the tree-shaped circuit (which is the sum of the sizes of the formulas $\varphi_i$), it follows that a tree-shaped 
circuit of size $N$ can be evaluated in space $O(\log^2 N)$. 
For the special case $m = O(1)$ (which is used for the bisimulation problem) we give 
an alternative proof below. This proof prepares our handling
of trees in term representation in the next section (proof of Theorem \ref{thm:eval-nc1}).

\label{alternative}

\begin{proof}[Alternative proof of Theorem~\ref{thm:eval-logspace} for constant width.]
	Let $N$ be the size of $\S$ and $m \in O(1)$ its width.
	We will construct from $\S$ in logspace an equivalent polynomially sized circuit with constant tree-width,
	which can be seen as a partial tree unfolding of the circuit corresponding to $\S$.
	By Theorem~\ref{theo-coro-elberfeld} the resulting circuit can be evaluated in logspace.
	
	As before, we view the Boolean formulas $\varphi_i$ as node-labelled trees.
	For simplicity we assume that all $\varphi_i$ have at least size two, which can be
	ensured by replacing $\varphi_i$ by $\varphi_i \wedge 1$, so that $\S$ contains no ``chain rules''. Moreover, we assume that
	the trees $\varphi_i$ have disjoint node sets.
	Let $x_i$ be an inner node of $T_\S$ whose heavy child is $x_k$.
	Inductively we define a circuit $C_i$ as follows: We take the formula $\varphi_i$, viewed as a tree,
	and merge all $x_k$-labelled leaves into a single node. Note that $x_j$-labelled leaves for $j \neq k$ are not merged.
	Then we insert into each leaf labelled by some variable $x_j$ a copy of the circuit $C_j$.
	Finally let $C$ be $C_1$, which clearly evaluates to the same truth value as $\S$.
	Figure~\ref{fig:partial-unfolding} shows the circuit resulting from the circuit $C_\sim(T_1, T_2)$ on the right in Figure~\ref{fig:tree-bisimulation-circuit}.
	
	Note that the number of copies of $C_i$ in $C$ is bounded by $m^{\ell_i}$ where $\ell_i$ is the number
	of light nodes on the path from $x_1$ to $x_i$ in $T_\S$. Since $\ell_i \leq \log n$,
	$C$ has size at most $m^{\log n} \cdot N$, which is bounded by $n^{O(1)} \cdot N$ since 
	$m$ is a constant.
	
	Furthermore $C$ can be computed in logspace from $\S$.
	To make this explicit, we introduce a naming scheme for the gates in $C$.
	The set \label{addr} $\mathrm{Addr}(x_i)$ (addresses for the copies of $x_i$) contains finite words over the alphabet $\{1, \dots, m\}$ defined inductively:
	We set $\mathrm{Addr}(x_1) = \{ \varepsilon \}$ for the root $x_1$.
	If $x_j$ is the heavy child of $x_i$ in $T_\S$, we set $\mathrm{Addr}(x_j) = \mathrm{Addr}(x_i)$.
	If $x_j$ is a light child of $x_i$ in $T_\S$, we set $\mathrm{Addr}(x_j) = \mathrm{Addr}(x_i) \cdot \{ 1, \dots, k \}$
	where $k \le m$ is the number of occurrences of $x_j$ in $\varphi_i$.
	The sets $\mathrm{Addr}(x_i)$ contain words of length $O(\log n)$ over the constant sized alphabet $\{1, \dots, m\}$. Moreover, given a word
	of length $O(\log n)$ over $\{1, \dots, m\}$ we can easily check in logspace whether it belongs to $\mathrm{Addr}(x_i)$ by traversing the path
	from $x_i$ to the root of $T_\S$.
	Now we can define the circuit $C$ over the gate set
	\[
	V(C) =  \bigcup_{i=1}^n \{ (u,a) \mid a \in \mathrm{Addr}(x_i), ~ u \text{ is a non-input node of the tree $\varphi_i$}\} .
	\]
	For every edge $(u,u')$ of $\varphi_i$ ($1 \le i \le n$) and
	$a \in \mathrm{Addr}(x_i)$, we add the following edges to $E(C)$:
	\begin{enumerate} \label{edges}
		\item If $u'$ is a non-input gate, add the edge $((u,a),(u',a))$.
		\item If $u'$ is labelled by $x_j$, let $v$ be the root of $\varphi_j$ and add the edge
		\begin{enumerate}
			\item $((u,a),(v,a \cdot d))$ if $x_j$ is a light child of $x_i$ and $u'$ is its $d$-th occurrence of $x_j$ in $\varphi_i$,
			\item $((u,a),(v,a))$ if $x_j$ is the heavy child of $x_i$. \label{heavy-edge}
		\end{enumerate}
	\end{enumerate}
	The labels of the gates in $C$ are inherited from the formula trees $\varphi_i$.
	Note that a pair $(u,a)$ from $V(C)$ can be stored in logspace. Moreover, whether a pair belongs to $V(C)$
	and whether a pair of nodes from $V(C)$ belongs to $E(C)$ can be checked in logspace. Hence, the circuit
	$C$ can be constructed in logspace.	
	
	Finally we show that the tree-width of $C$ is at most $2$.
	Let us call an edge in $C$ {\em crossing}
	if it is of type \ref{heavy-edge} in the above definition of $C$
	and $u'$ is the $d$-th occurrence of $x_j$ in $\varphi_i$ for some $d > 1$.
	In Figure~\ref{fig:partial-unfolding} the crossing edges are drawn as dotted lines.
	Notice that for each address $a$ there is at most one crossing edge $((u,a),(v,a)) \in E(C)$.
	Consider the subgraph $T$ of $C$ obtained by removing all crossing edges.
	The resulting subgraph $T$ is indeed a tree, on which we define a tree decomposition $(T,\beta)$
	\label{tree-width-2-deco} of $C$.
	First let $(T,\beta')$ be the ``canonical'' tree decomposition of $T$ of width 1,
	i.e. $\beta'(u,a)$ contains $(u,a)$ and its parent node in $T$ (if existent).
	To obtain $\beta$ from $\beta'$ we add for each address $a$ the crossing edge $((u,a),(v,a))$ (if existent)
	to all bags $\beta(x,a)$.
	\label{bags}
	One can verify that $(T,\beta)$ is indeed a valid tree decomposition,
	which has width at most 2.
\end{proof} 

Theorem~\ref{thm:eval-logspace} applied to the tree-shaped circuit defined by \eqref{bisi-formula}
immediately yields Theorem~\ref{thm-bisi-logspace} for the problem $\Bisim(\mathrm{Trees})$ when
trees are given in pointer representation.

\begin{figure}[t]
	\centering
	
	\begin{tikzpicture}
	
	\tikzset{level distance=30pt, sibling distance=10pt}
	\tikzset{every tree node/.style={gate}}
	\Tree
	[. \node {$\wedge$};
	[. \node {$\vee$};
	[. \node (b) {$\wedge$};
	[. \node {$\vee$};
	[. \node (d) {$1$}; ]
	[. \node {$1$}; ]
	]
	[. \node (c) {$\vee$}; ]
	[. \node {$\vee$};
	[. \node {$1$}; ]
	]
	]
	[. \node {0}; ]
	]
	[. \node {$\vee$};
	[. \node {$\wedge$};
	[. \node {$\vee$};
	[. \node (f) {$1$}; ]
	[. \node {$1$}; ]
	]
	[. \node (e) {$\vee$}; ]
	[. \node {$\vee$};
	[. \node {$1$}; ]
	]
	]
	[. \node {0}; ]
	]
	[. \node (a) {$\vee$};
	[. \node {$\wedge$};
	[. \node {$\vee$};
	[. \node (h) {$1$}; ]
	[. \node {$1$}; ]
	]
	[. \node (g) {$\vee$}; ]
	[. \node {$\vee$};
	[. \node {$1$}; ]
	]
	]
	]
	[. \node {$\vee$};
	[. \node {0}; ]
	[. \node {0}; ]
	]
	]
	
	\draw[arrow,densely dotted] (a) .. controls +(south west:1) and +(north east:1) .. (b.20);
	\draw[arrow,densely dotted] (c) .. controls +(south west:.6) and +(north east:.8) .. (d);
	\draw[arrow,densely dotted] (e) .. controls +(south west:.6) and +(north east:.8) .. (f);
	\draw[arrow,densely dotted] (g) .. controls +(south west:.6) and +(north east:.8) .. (h);
	
	\end{tikzpicture}
	
	\caption{The partial unfolding of the tree-shaped circuit $C_\sim(T_1,T_2)$ from Figure~\ref{fig:tree-bisimulation-circuit}.}
	\label{fig:partial-unfolding}
\end{figure}
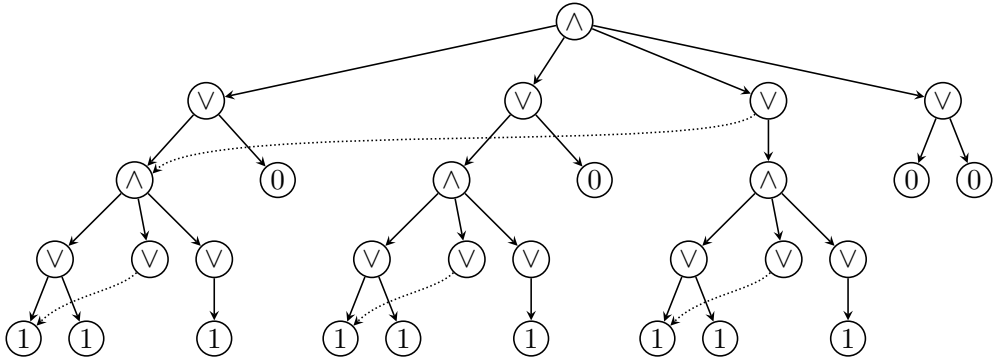

\subsubsection{Evaluating tree-shaped circuits given in term representation}
 \label{terms}

In this section we will prove Theorem~\ref{thm:bisi-nc1} for the problem $\Bisim(\mathrm{Trees})$ when
trees are given in term representation.
For that we prove an $\NC^1$-version of Theorem~\ref{thm:eval-logspace} for the case that $m$ is constant,
which uses the $\NC^1$-part of Theorem~\ref{theo-coro-elberfeld}, where a tree decomposition
in term representation is part of the input.
Here we require that the tree-shaped circuit $\S= (x_i = \varphi_i)_{1 \leq i \leq n}$ must be given together with the underlying tree $T_\S$ in term representation.
We also assume that the Boolean formulas $\varphi_i$ are given in term representation.
Recall the ancestor representation of a tree from Section~\ref{sec-prel} and that it can be transformed into the term
representation of a tree in  $\TC^0$ and vice versa.

\begin{thm}
	\label{thm:eval-nc1}
	For every fixed $m \in \N$, one can evaluate in $\NC^1$
	a given tree-shaped circuit $\S= (x_i = \varphi_i)_{1 \leq i \leq n}$ (with all 
	$\varphi_i$ given in term representation) of width at most $m$ that 
	is given together with the term representation of the tree $T_\S$. 
\end{thm}

\begin{proof}
	Let $\S = (x_i = \varphi_i)_{1 \leq i \leq n}$ be a tree-shaped circuit of width at most $m$, where
	every $\varphi_i$ is given in term representation. Moreover, we assume to have the term (or ancestor)
	representation of the tree $T_\S$.
	We will show how to compute in $\TC^0$ the partial unfolding $C$ of $\S$ and the 
	tree decomposition of $C$ in ancestor representation from the alternative proof of Theorem \ref{thm:eval-logspace} on page \pageref{alternative}.
	
	Let $N$ be the size of $\S$.
	For each node $x_i$ of $T_\S$ we can compute in $\TC^0$ its depth (by counting ancestors) and
	the size of the subtree below $x_i$ (by counting descendants).
	Hence, we can also compute the heavy child of every inner node $x_i$ in $T_\S$.
	Additionally, if $x_i$ has a parent node $x_j$, we can compute the number $d(x_i) \in \{1, \dots, m\}$ of occurrences
	of $x_i$ in $\varphi_j$.
	We transform all formulas $\varphi_i$ of $\S$ into ancestor representation in $\TC^0$ where we can assume
	that the encodings of all nodes have length $O(\log N)$.
	Also we can ensure that all formulas have at least size two.
	
	Recall that every node $u$ of a formula $\varphi_i$ has multiple copies $(u,a)$ in $C$ where
	$a \in \mathrm{Addr}(x_i)$ is an address of length $O(\log n)$ defined on page \pageref{addr}.
	Given a string $a \in \{1, \dots, m\}^*$ of length $O(\log n)$ and a node $x_i$ we can verify in $\TC^0$
	whether $a \in \mathrm{Addr}(x_i)$:
	From the ancestor representation of $T_\S$ we can compute the sequence of all light nodes
	$x_{i_1} \prec x_{i_2} \prec \dots \prec x_{i_k}$ in $T_\S$ on the path from the root to $x_i$.
	This can be done by sorting all light ancestors of $x_i$ by their depth in ascending order.
	It is known that sorting $n$ numbers with $n$ bits each is in $\TC^0$ \cite{Vol99}.
	Then we have $a_1 \cdots a_k \in \mathrm{Addr}(x_i)$ if and only if $a_j \le d(x_{i_j})$ for all $j \in \{1, \dots, k\}$.
	Hence, we can also encode all gates of the partial unfolding $C$ by strings of length $O(\log N)$ and can compute $V(C)$ in $\TC^0$.
	With the previous preparation the edge relation $E(C)$ can be computed in $\AC^0$ using the definition on page \pageref{edges}.
	
	It remains to show that we can compute in $\TC^0$ the ancestor representation of the 
	width-2 tree decomposition $(T,\beta)$ of $C$ from page~\pageref{tree-width-2-deco}.
	We set $V(T) = V(C)$. Let $(u,a), (u',a') \in V(T)$ be nodes where $u$ (resp., $u'$) belongs to $\varphi_i$ (resp., $\varphi_j$).
	Then $(u,a)$ is an ancestor of $(u',a')$ in $T$ if and only if $x_i \preceq x_j$ in $T_\S$, $a$ is a prefix of $a'$ and the following holds:
	\begin{enumerate}
		\item If $x_i = x_j$, then $a = a'$ and $u \preceq u'$ in $\varphi_i$.
		\item If $x_i \prec x_j$ in $T_\S$, let $x_k$ be the unique child of $x_i$ which is an ancestor of $x_j$.
		\begin{enumerate}
			\item If $x_k$ is a light node, let $d \in \{1, \dots, m\}$ be the number in $a'$ at position $|a|+1$.
			Then the $d$-th occurrence of $x_k$ in $\varphi_i$ is a descendant of $u$.
			\item If $x_k$ is a heavy node, then the first occurrence of $x_k$ in $\varphi_i$ is a descendant of $u$.
		\end{enumerate}
	\end{enumerate}
	The last condition forbids those edges that were deleted when constructing the tree $T$ from the partial unfolding $C$
	(the dotted edges in Figure~\ref{fig:partial-unfolding}).
	The bag-function $\beta$ can be computed straightforwardly in $\TC^0$ using its definition on page~\pageref{bags}.
	By Theorem~\ref{theo-coro-elberfeld} we can evaluate $C$ in $\NC^1$, which concludes the proof.
\end{proof}

Finally we can apply Theorem~\ref{thm:eval-nc1} to prove Theorem~\ref{thm:bisi-nc1}:

\begin{proof}[Proof of Theorem~\ref{thm:bisi-nc1}]
	First we convert the term representations of $T_1$ and $T_2$ into ancestor representations in $\TC^0$.
	We can compute the tree-shaped circuit $\S$ corresponding to the circuit $C_\sim(T_1,T_2)$
	and an ancestor representation of $T_\S$ in $\TC^0$.
	The set of variables
	\[
    	\{ x_{u,v} \mid u \in V(T_1), \, v \in V(T_2), \, \depth(u) = \depth(v) \}
	\]
	can be clearly computed in $\TC^0$ since for a given node of a tree in ancestor representation
	one can compute its depth by counting ancestors.
	The term representations of the formulas $\varphi_{u,v}$ in \eqref{bisi-formula} can then be computed in $\AC^0$.
	The ancestor representation of $T_\S$ is also $\AC^0$-computable, since
	$x_{u,v}$ is an ancestor of $x_{u',v'}$ if and only if
	$u$ is an ancestor of $u'$ and $v$ is an ancestor of $v'$.
	By Theorem~\ref{thm:eval-nc1} we can evaluate $\S$ in $\NC^1$.
\end{proof}

\subsubsection{Lower Bounds} \label{bisi-lower-bound}

In this section we prove matching lower bounds
for the upper bounds from Theorem~\ref{thm-bisi-logspace} and ~\ref{thm:bisi-nc1}. 

\begin{thm} \label{thm-lower-bound-L-NC1}
	{\rm $\Bisim$(Unlabelled-Trees)} is \L-hard if the trees are given as pointer structures and
	$\NC^1$-hard if they are given in term representation (both with respect to many-one $\AC^0$-reductions).
\end{thm}

Before we prove Theorem~\ref{thm-lower-bound-L-NC1} let us first show the following lemma:

\begin{lem} \label{lemma-unlabeling-in-AC0}
	{\rm $\Bisim$(Trees)} is many-one $\AC^0$-reducible to {\rm $\Bisim$(Unlabelled-Trees)} in both term and pointer representation.
	%The bisimulation problem for edge-labelled trees in term representation (resp., pointer representation) is $\AC^0$-reducible
	%to the bisimulation problem for unlabelled trees in term representation (resp., pointer representation).
\end{lem}

\begin{proof}
	We only show the lemma for the term representation; the same construction also works for the pointer representation.
	In \cite{Srba01} Srba presents a reduction from the bisimulation problem for 
	edge-labelled graphs to the bisimulation problem for unlabelled graphs.
	In fact, this construction transforms trees into trees.
	We slightly modify the reduction to ensure $\AC^0$-computability
	and assume that there are only two labels, say $a$ and $b$
	(which is the case for the trees constructed in the proof of Theorem \ref{thm-lower-bound-L-NC1}).
	
	Consider a tree $T$ with edge labels $a$ and $b$. 
	First every labelled edge of $T$ is subdivided into two edges. 
	In Figure~\ref{fig-srba} (middle tree), the new node added for an $x$-labelled edge ($x \in \{a,b\}$)
	is labelled with $x$. 
	To distinguish the original nodes from the new nodes,
	we attach to each original node two paths of length $3$. To each new node we attach one of two small trees
	depending on the label of the original edge that is represented by the new node, see the right tree in Figure~\ref{fig-srba}.
	Let us denote the resulting unlabelled tree with $\mathsf{ul}(T)$.
	It is not hard to prove that two labelled trees $T_1$ and $T_2$ are bisimilar if and only if 
	$\mathsf{ul}(T_1)$ and $\mathsf{ul}(T_2)$ are bisimilar.
	The proof is  basically given in \cite{Srba01}.
	
	\begin{figure}[t]
		\centering
		\tikzset{every node/.style={circle,draw,inner sep = 1.5pt}}
		\begin{tikzpicture}	
		\node (a) {};
		\node[below left = 4cm and 1cm of a] (b) {};
		\node[below right = 4cm and 1cm of a] (c) {};
		
		\draw[arrow] (a) edge node [above left, draw=none] {$a$} (b);
		\draw[arrow] (a) edge node [above right, draw=none] {$b$} (c);
		\end{tikzpicture}
		\hspace{.5cm}
		\begin{tikzpicture}	
		\node (a) {};
		\node[below left = 2cm and .5cm of a,label=left:{$a$}] (b) {};
		\node[below right = 2cm and .5cm of a,label=right:{$b$}] (c) {};
		\node[below left = 2cm and .5cm of b] (d) {};
		\node[below right = 2cm and .5cm of c] (e) {};
		
		\draw[arrow] (a) -- (b);
		\draw[arrow] (a) -- (c);
		\draw[arrow] (b) -- (d);
		\draw[arrow] (c) -- (e);
		\end{tikzpicture}
		\hspace{.5cm}
		\begin{tikzpicture}	
		\node (a) {};
		\node[below left = 2cm and 1cm of a] (b) {};
		\node[below left = 2cm and 1cm of b] (c) {};
		\node[below right = 2cm and 1cm of a] (d) {};
		\node[below right = 2cm and 1cm of d] (e) {};
		
		\node[right = .3cm of a] (a1) {};
		\node[right = .3cm of a1] (a2) {};
		\node[right = .3cm of a2] (a3) {};
		\node[left = .3cm of a] (a4) {};
		\node[left = .3cm of a4] (a5) {};
		\node[left = .3cm of a5] (a6) {};
		
		\node[above left = .2cm and .5cm of b] (b1) {};
		\node[left = .5cm of b] (b2) {};
		\node[below left = .2cm and .5cm of b] (b3) {};
		\node[above right = .2cm and .5cm of b] (b4) {};
		\node[right = .5cm of b] (b5) {};
		\node[below right = .2cm and .5cm of b] (b6) {};
		
		\node[right = .3cm of c] (c1) {};
		\node[right = .3cm of c1] (c2) {};
		\node[right = .3cm of c2] (c3) {};
		\node[left = .3cm of c] (c4) {};
		\node[left = .3cm of c4] (c5) {};
		\node[left = .3cm of c5] (c6) {};
		
		\node[above left = .2cm and .8cm of d] (d1) {};
		\node[left = .4cm of d] (d2) {};
		\node[below left = .2cm and .8cm of d] (d3) {};
		\node[above right = .2cm and .5cm of d] (d4) {};
		\node[right = .5cm of d] (d5) {};
		\node[below right = .2cm and .5cm of d] (d6) {};
		
		\node[right = .3cm of e] (e1) {};
		\node[right = .3cm of e1] (e2) {};
		\node[right = .3cm of e2] (e3) {};
		\node[left = .3cm of e] (e4) {};
		\node[left = .3cm of e4] (e5) {};
		\node[left = .3cm of e5] (e6) {};
		
		\draw[arrow] (a) -- (b);
		\draw[arrow] (b) -- (c);
		\draw[arrow] (a) -- (d);
		\draw[arrow] (d) -- (e);
		
		\draw[arrow] (a) -- (a1);
		\draw[arrow] (a1) -- (a2);
		\draw[arrow] (a2) -- (a3);
		\draw[arrow] (a) -- (a4);
		\draw[arrow] (a4) -- (a5);
		\draw[arrow] (a5) -- (a6);
		
		\draw[arrow] (b) -- (b1);
		\draw[arrow] (b) -- (b2);
		\draw[arrow] (b) -- (b3);
		\draw[arrow] (b) -- (b4);
		\draw[arrow] (b) -- (b5);
		\draw[arrow] (b) -- (b6);
		
		\draw[arrow] (c) -- (c1);
		\draw[arrow] (c1) -- (c2);
		\draw[arrow] (c2) -- (c3);
		\draw[arrow] (c) -- (c4);
		\draw[arrow] (c4) -- (c5);
		\draw[arrow] (c5) -- (c6);
		
		\draw[arrow] (d2) -- (d1);
		\draw[arrow] (d) -- (d2);
		\draw[arrow] (d2) -- (d3);
		\draw[arrow] (d) -- (d4);
		\draw[arrow] (d) -- (d5);
		\draw[arrow] (d) -- (d6);
		
		\draw[arrow] (e) -- (e1);
		\draw[arrow] (e1) -- (e2);
		\draw[arrow] (e2) -- (e3);
		\draw[arrow] (e) -- (e4);
		\draw[arrow] (e4) -- (e5);
		\draw[arrow] (e5) -- (e6);
		\end{tikzpicture}
		\caption{From labelled to unlabelled trees}
		\label{fig-srba}
	\end{figure}
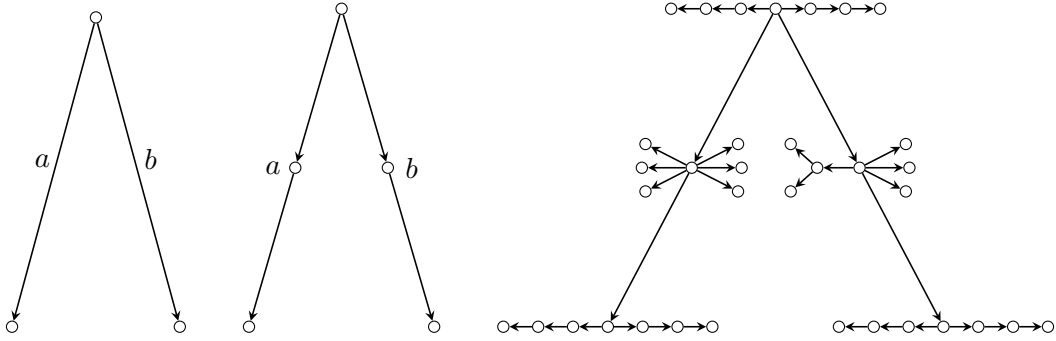

	It remains to prove that the term representation of $\mathsf{ul}(T)$ can be computed in $\AC^0$
	from the term representation of $T$.	Consider the term representation $t$ of $T$.
	Recall that we identify the opening brackets in $t$ with the nodes of $T$.
	We assume that the term representation $t$ contains the following three opening bracket types:
	$\auf_a$ and $\auf_b$, which represent nodes with an incoming $a$-labelled (resp., $b$-labelled) edge,
	and $\auf$ for the root.
	For example, $t = \auf \auf_a \zu \auf_b \zu \zu$ is a term representation for the left tree in Figure \ref{fig-srba}.
	The transformation $T \mapsto \mathsf{ul}(T)$ can be described by two isometric homomorphisms.
	A homomorphism $h\colon \Sigma^* \to \Gamma^*$ is {\em isometric} if there is an $\ell \geq 1$ such that
	$|h(c)| = \ell$ for all $c \in \Sigma$. In  \cite{LangeM98} it is shown that 
	for a given isometric homomorphism $h\colon \Sigma^* \to \Gamma^*$ and a word $w \in \Sigma^*$
	one can compute $h(w)$ in $\AC^0$.\footnote{If the homomorphism is fixed, this is even possible in $\NC^0$. Moreover, if 
		the homomorphism is not isometric then the problem is $\TC^0$-complete \cite{LangeM98}.}
	
	We will proceed in two steps. 
	Define the isometric homomorphism $h_1\colon \{ \auf_a, \auf_b, \zu \}^* \to \{ \auf_a, \auf_b, \auf, \zu, \Zu \}^*$ by:
	$$
	\auf_a \mapsto \auf_a \auf  \qquad\qquad \auf_b \mapsto \auf_b \auf \qquad\qquad \zu \mapsto ~ \zu \Zu
	$$
	Let $u \in \{ \auf_a, \auf_b, \zu \}^*$ be the word such that $t = \auf u \zu$ and consider the string
	$\auf h_1(u) \zu$. 
	Formally,  it is not a term representation (since we have two types of closing brackets). Nevertheless, it describes the tree obtained from $T$
	by subdividing every edge and labelling each new node with the former edge label.
	For example $t = \auf \auf_a \zu \auf_b \zu \zu$ is transformed into
	$\auf h_1(u) \zu = \auf \auf_a \auf \zu \Zu \auf_b \auf \zu \Zu \zu$,
	which describes the node-labelled tree in Figure \ref{fig-srba}.
	The second isometric homomorphism
	$h_2\colon \{ \auf_a, \auf_b, \auf, \zu, \Zu \}^* \to \{ \auf, \zu \}^*$ is defined by:
	\begin{itemize}
		\item $\auf \mapsto \auf \auf\auf\auf\zu\zu\zu$ (an opening bracket followed by a path of length $3$)
		\item $\zu \mapsto \auf\auf\auf\zu\zu\zu \zu$    (a path of length $3$ followed by a closing bracket)
		\item $\Zu \mapsto \auf\zu\auf\zu\auf\zu \zu$ (3 leaves followed by a closing bracket)
		\item $\auf_a \mapsto \auf \auf\zu\auf\zu\auf\zu$ (an opening bracket followed by 3 leaves)
		\item $\auf_b \mapsto \auf \auf\auf\zu\auf\zu\zu$ (an opening bracket followed by the tree  \!\!\!\!\!\!
		\mbox{
			\tikzset{every node/.style={circle,draw,inner sep = 1pt,minimum size = 1pt}}
			\begin{tikzpicture}	
			\node (a) {};
			\node[below left = 0.18cm and 0.09cm of a] (b) {};
			\node[below right = 0.18cm and 0.09cm of a] (c) {};
			\draw[->,>=stealth] (a) -- (b);
			\draw[->,>=stealth] (a) -- (c);
			\end{tikzpicture}})
	\end{itemize}
	Then,  
	the string $h_2( \auf h_1(u) \zu)$ is indeed a term representation for the desired unlabelled tree~$\mathsf{ul}(T)$.
\end{proof}

\begin{proof}[Proof of Theorem~\ref{thm-lower-bound-L-NC1}]
	By Lemma~\ref{lemma-unlabeling-in-AC0} it suffices to show the lower bounds for edge-labelled trees.
	We reuse the proofs from  \cite{JennerKMT03}, where it is shown that 
	the tree isomorphism problem is \L-hard ($\NC^1$-hard, respectively) with respect to $\AC^0$-reductions
	if the the trees are given as pointer structures (in term representation, respectively).
	Let us start with the bisimulation problem for trees given in pointer representation.
	Here, Jenner et al. reduce from the \L-complete reachability problem on paths, i.e.,
	the question whether for a given directed path graph $G$ and two nodes $v_i, v_j \in V(G)$, there is a path from $v_i$ to $v_j$ \cite{Ete97}.	
	Without loss of generality, $v_i$ and $v_j$ are distinct and have successors $v_{i+1}$ and $v_{j+1}$, respectively.
	
	Consider the tree with a root node which has two copies of $G$ as direct subtrees.
	We refer to nodes of the two copies by $v_1, \dots, v_n$ and $v_1', \dots, v_n'$, respectively.
	Additionally we replace the edge $(v_i,v_{i+1})$ by the new edge $(v_i',v_{i+1})$.
	Now let $T_1$ (resp., $T_2$) be the tree where the edge $(v_j,v_{j+1})$ (resp., $(v_j',v_{j+1}')$)
	is labelled by a symbol $a$ (all unlabelled edges are assumed to be labelled with a symbol $b \neq a$). 
	Clearly, $T_1$ and $T_2$ can be computed in $\AC^0$ from $G$.
	There is a path from $v_i$ to $v_j$ in $G$ if and only
	if $T_1$ and $T_2$ are bisimilar.
	See Figure~\ref{fig:jenner-logspace} for an illustration of the reduction.
	
	\begin{figure}[t]
		\tikzstyle{invisible} = [draw=none,inner sep=0,minimum size=0]
		\tikzset{every node/.style={circle,draw,inner sep = 1.5pt}}
		\begin{minipage}{.45\linewidth}
			\centering
			\begin{tikzpicture}	
			\node (r) {};
			\node[above right = 1em of r] (a0) {};
			\node[below right = 1em of r] (a1) {};
			\node[right = 3em of a0,label=above:{\footnotesize $v_i$}] (u0) {};
			\node[right = 3em of a1,label=below:{\footnotesize $v_i'$}] (u1) {};
			\node[right = 1em of u0] (uu0) {};
			\node[right = 1em of u1] (uu1) {};
			\node[right = 4em of uu0,label=above:{\footnotesize $v_j$}] (v0) {};
			\node[right = 1em of v0] (vv0) {};
			\node[right = 3.5em of vv0,invisible] (b0) {};
			\node[right = 4em of uu1,label=below:{\footnotesize $v_j'$}] (v1) {};
			\node[right = 1em of v1] (vv1) {};	
			\node[right = 3.5em of vv1,invisible] (b1) {};
			
			\draw[arrow] (r) -- (a0);
			\draw[arrow] (r) -- (a1);
			\draw[arrow] (a0) -- (u0);
			\draw[arrow] (a1) -- (u1);
			\draw[arrow] (u1) -- (uu1);
			\draw[arrow] (u1) -- (uu0);
			\draw[arrow] (uu0) -- (v0);
			\draw[arrow] (v0) edge node[draw=none,below] {\footnotesize $a$} (vv0);
			\draw[arrow] (vv0) -- (b0);
			\draw[arrow] (uu1) -- (v1);
			\draw[arrow] (v1) -- (vv1);
			\draw[arrow] (vv1) -- (b1);
			
			\node[below = 6em of r]  (r') {};
			\node[above right = 1em of r'] (a0') {};
			\node[below right = 1em of r'] (a1') {};
			\node[right = 3em of a0',label=above:{\footnotesize $v_i$}] (u0') {};
			\node[right = 3em of a1',label=below:{\footnotesize $v_i'$}] (u1') {};
			\node[right = 1em of u0'] (uu0') {};
			\node[right = 1em of u1'] (uu1') {};
			\node[right = 4em of uu0',label=above:{\footnotesize $v_j$}] (v0') {};
			\node[right = 1em of v0'] (vv0') {};
			\node[right = 3.5em of vv0',invisible] (b0') {};
			\node[right = 4em of uu1',label=below:{\footnotesize $v_j'$}] (v1') {};
			\node[right = 1em of v1'] (vv1') {};	
			\node[right = 3.5em of vv1',invisible] (b1') {};
			
			\draw[arrow] (r') -- (a0');
			\draw[arrow] (r') -- (a1');
			\draw[arrow] (a0') -- (u0');
			\draw[arrow] (a1') -- (u1');
			\draw[arrow] (u1') -- (uu1');
			\draw[arrow] (u1') -- (uu0');
			\draw[arrow] (uu0') -- (v0');
			\draw[arrow] (v0') -- (vv0');
			\draw[arrow] (vv0') -- (b0');
			\draw[arrow] (uu1') -- (v1');
			\draw[arrow] (v1') edge node[draw=none,above] {\footnotesize $a$} (vv1');
			\draw[arrow] (vv1') -- (b1');
			\end{tikzpicture}
			\subcaption{$T_1$ and $T_2$ are bisimilar.}
		\end{minipage}
		\hfill
		\begin{minipage}{.45\textwidth}
			\centering
			\begin{tikzpicture}
			\node (r) {};
			\node[above right = 1em of r] (a0) {};
			\node[below right = 1em of r] (a1) {};
			\node[right = 3em of a0,label=above:{\footnotesize $v_j$}] (u0) {};
			\node[right = 3em of a1,label=below:{\footnotesize $v_j'$}] (u1) {};
			\node[right = 1em of u0] (uu0) {};
			\node[right = 1em of u1] (uu1) {};
			\node[right = 4em of uu0,label=above:{\footnotesize $v_i$}] (v0) {};
			\node[right = 1em of v0] (vv0) {};
			\node[right = 3.5em of vv0,invisible] (b0) {};
			\node[right = 4em of uu1,label=below:{\footnotesize $v_i'$}] (v1) {};
			\node[right = 1em of v1] (vv1) {};	
			\node[right = 3.5em of vv1,invisible] (b1) {};
			
			\draw[arrow] (r) -- (a0);
			\draw[arrow] (r) -- (a1);
			\draw[arrow] (a0) -- (u0);
			\draw[arrow] (a1) -- (u1);
			\draw[arrow] (u1) -- (uu1);
			\draw[arrow] (uu0) -- (v0);
			\draw[arrow] (u0) edge node[draw=none,below] {\footnotesize $a$} (uu0);
			\draw[arrow] (v1) -- (vv0);
			\draw[arrow] (vv0) -- (b0);
			\draw[arrow] (uu1) -- (v1);
			\draw[arrow] (v1) -- (vv1);
			\draw[arrow] (vv1) -- (b1);
			
			\node[below = 6em of r] (r') {};
			\node[above right = 1em of r'] (a0') {};
			\node[below right = 1em of r'] (a1') {};
			\node[right = 3em of a0',label=above:{\footnotesize $v_j$}] (u0') {};
			\node[right = 3em of a1',label=below:{\footnotesize $v_j'$}] (u1') {};
			\node[right = 1em of u0'] (uu0') {};
			\node[right = 1em of u1'] (uu1') {};
			\node[right = 4em of uu0',label=above:{\footnotesize $v_i$}] (v0') {};
			\node[right = 1em of v0'] (vv0') {};
			\node[right = 3.5em of vv0',invisible] (b0') {};
			\node[right = 4em of uu1',label=below:{\footnotesize $v_i'$}] (v1') {};
			\node[right = 1em of v1'] (vv1') {};	
			\node[right = 3.5em of vv1',invisible] (b1') {};
			
			\draw[arrow] (r') -- (a0');
			\draw[arrow] (r') -- (a1');
			\draw[arrow] (a0') -- (u0');
			\draw[arrow] (a1') -- (u1');
			\draw[arrow] (u1') edge node[draw=none,above] {\footnotesize $a$}  (uu1');
			\draw[arrow] (uu0') -- (v0');
			\draw[arrow] (u0') -- (uu0');
			\draw[arrow] (v1') -- (vv0');
			\draw[arrow] (vv0') -- (b0');
			\draw[arrow] (uu1') -- (v1');
			\draw[arrow] (v1') -- (vv1');
			\draw[arrow] (vv1') -- (b1');
			\end{tikzpicture}
			\subcaption{$T_1$ and $T_2$ are not bisimilar.}
		\end{minipage}
		
		\caption{The two possible forms of $T_1$, $T_2$ depending on whether $v_j$ is reachable from $v_i$ or not}
		\label{fig:jenner-logspace}
	\end{figure}
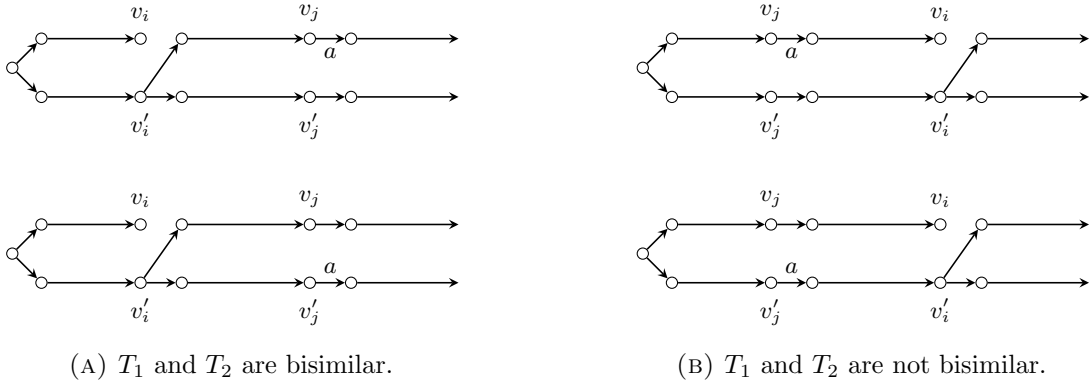
	
	Secondly, Jenner et al. present in \cite{JennerKMT03} an $\AC^0$-reduction from the $\NC^1$-complete evaluation problem of balanced
	Boolean expressions to the isomorphism problem for trees in term representation.
	They use the AND-gadget $T_\wedge(G_1,G_2,H_1,H_2)$ and the OR-gadget $T_\vee(G_1,G_2,H_1,H_2)$
	which are depicted in Figure~\ref{fig-lower-NC1}.
	Notice that for all trees $G_1,G_2,H_1,H_2$ the following holds, where $\sim$ can both mean bisimilarity and isomorphism:	
	\begin{figure}[t]
		\tikzset{level distance=25pt, sibling distance=5pt}
		\tikzset{every tree node/.style={draw, circle, inner sep = 1.5pt}}
		\begin{minipage}{.45\textwidth}
			\centering
			\begin{tikzpicture}
			\Tree
			[. \node {};
			\edge node[auto=right]{\footnotesize $a$};
			[. \node {};
			\edge node[auto=right]{\footnotesize $a$};
			[. \node[draw=none] {\footnotesize $G_1$}; ]
			[. \node[draw=none] {\footnotesize $G_2$}; ]
			]
			[. \node {};
			\edge node[auto=right]{\footnotesize $a$};
			[. \node[draw=none] {\footnotesize $H_1$}; ]
			[. \node[draw=none] {\footnotesize $H_2$}; ]
			]
			]
			\end{tikzpicture}
			\begin{tikzpicture}
			\Tree
			[. \node {};
			\edge node[auto=right]{\footnotesize $a$};
			[. \node {};
			\edge node[auto=right]{\footnotesize $a$};
			[. \node[draw=none] {\footnotesize $H_1$}; ]
			[. \node[draw=none] {\footnotesize $H_2$}; ]
			]
			[. \node {};
			\edge node[auto=right]{\footnotesize $a$};
			[. \node[draw=none] {\footnotesize $G_1$}; ]
			[. \node[draw=none] {\footnotesize $G_2$}; ]
			]
			]
			\end{tikzpicture}
			\subcaption{AND-gadget $T_\wedge$}
		\end{minipage}
		\hfill
		\begin{minipage}{.45\textwidth}
			\centering
			\begin{tikzpicture}
			\Tree
			[. \node {};
			[. \node {};
			\edge node[auto=right]{\footnotesize $a$};
			[. \node[draw=none] {\footnotesize $G_1$}; ]
			[. \node[draw=none] {\footnotesize $G_2$}; ]
			]
			[. \node {};
			\edge node[auto=right]{\footnotesize $a$};
			[. \node[draw=none] {\footnotesize $H_1$}; ]
			[. \node[draw=none] {\footnotesize $H_2$}; ]
			]
			]
			\end{tikzpicture}
			\begin{tikzpicture}
			\Tree
			[. \node {};
			[. \node {};
			\edge node[auto=right]{\footnotesize $a$};
			[. \node[draw=none] {\footnotesize $G_1$}; ]
			[. \node[draw=none] {\footnotesize $H_2$}; ]
			]
			[. \node {};
			\edge node[auto=right]{\footnotesize $a$};
			[. \node[draw=none] {\footnotesize $H_1$}; ]
			[. \node[draw=none] {\footnotesize $G_2$}; ]
			]
			]
			\end{tikzpicture}
			\subcaption{OR-gadget $T_\vee$}
		\end{minipage}
		\caption{The trees for the $\NC^1$ lower bound}
		\label{fig-lower-NC1}
	\end{figure}
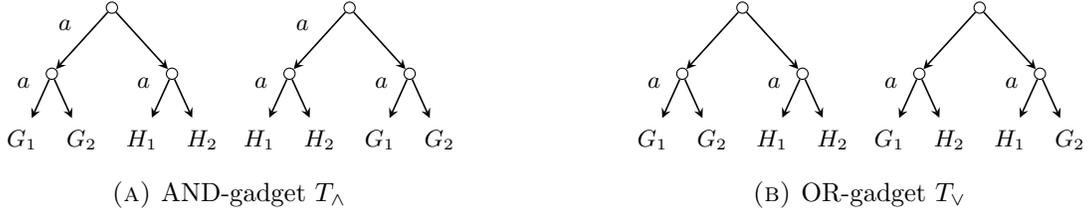
	\begin{itemize}
		\item $G_1 \sim H_1$ and $G_2 \sim H_2 \iff T_\wedge(G_1,G_2,H_1,H_2) \sim T_\wedge(H_1,H_2,G_1,G_2)$, and
		\item $G_1 \sim H_1$ or $G_2 \sim H_2 \iff T_\vee(G_1,G_2,H_1,H_2) \sim T_\vee(G_1,H_2,H_1,G_2)$.
	\end{itemize}
	\medskip
	Using the fact that the AND-gadget and the OR-gadget have the same tree structure (if we ignore labels), 
	one can show that the term representations for the resulting trees can be computed in $\AC^0$ from the balanced Boolean expression;
	see the arguments in  \cite{JennerKMT03}. This yields an $\AC^0$-reduction from 
	the evaluation problem of balanced Boolean expressions to the bisimulation problem for trees in term representation.
\end{proof}

\subsection{Bisimulation between trees and general graphs}

As remarked in Section~\ref{sec-bisi+sim}, the bisimulation problem is \P-complete 
for the class of all finite graphs. For trees we have shown that the complexity reduces to 
logspace, respectively $\NC^1$, depending on the representation of the input trees. This
leaves the question for the complexity of the bisimulation problem, when one of the input
graphs is a tree, and the other graph is not restricted. In this section we show that this problem
belongs to $\AC^1$ and that it is $\NL$-hard, but we are not able to prove matching bounds.
Let us start with the lower bound.

\begin{thm}
	$\Bisim${\rm (Unlabelled-Trees, Unlabelled-Graphs)} is $\NL$-hard.
\end{thm}

\begin{proof}
	We reduce from the following $\NL$-complete problem: Given a graph $G = (V,E)$, two vertices $s,t \in V$ and a number $k \in \N$,
	does there exist a path of length exactly $k$ from $s$ to $t$?
	
	Let $H$ be the dag defined on $V(H) = V \times \{1, \ldots, k+1\} \cup \{ * \}$ with the edge set
	\[
		E(H) = \{ ((u,i),(v,i+1)) \mid (u,v) \in E, \, 1 \le i \le k \} \cup \{ ((t,k+1),*) \}.
	\]
	Notice that $t$ is reachable from $s$ in exactly $k$ steps if and only if the node $(s,1)$
	is not bisimilar to a path of length $k$.
	Clearly, $H$ and the path graph can be constructed in logspace from $G$.  The result follows, since $\NL$
	is closed under complement.
\end{proof}
In the rest of the section, we show that $\Bisim${\rm (Trees, Graphs)} belongs to $\AC^1$.
Using the reduction in \cite{Srba01} (see also the proof of Lemma~\ref{lemma-unlabeling-in-AC0})
we can assume that the input tree and graph are unlabelled.
For a number $d \in \N$ we write $u \arrow{d} u'$ if there exists a path of length $d$ from $u$ to $u'$.
Let $G = (V,E)$ be a finite unlabelled directed graph, on which we define approximants $\approx_k$ of the bisimulation equivalence $\sim_G$.
First we set $u \approx_0 v$ for all $u, v \in V$. Inductively, we define $u \approx_{k+1} v$ if for all $d \in \N$ we have
\begin{itemize}
\item for all $u \arrow{d} u'$ there exists $v \arrow{d} v'$ such that $u' \approx_k v'$, and
\item for all $v \arrow{d} v'$ there exists $u \arrow{d} u'$ such that $u' \approx_k v'$.
\end{itemize}
Notice that $\approx_i \, \supseteq \, \approx_{i+1}$ for all $i \in \N$. Hence there exists $n \in \N$ such that
$\approx_n \, = \, \approx_{n+i}$ for all $i \in \N$ and $\approx_n$ coincides with the bisimulation equivalence $\sim$.
Before considering the bisimulation problem between trees and graphs, we show that the approximants $\approx_k$
converge after $O(\log n)$ rounds on trees of size $n$. For nodes $u,v \in V$ we define
\[
	r(u,v) = \sup \{ k \in \N \mid u \approx_k v \} \in \N \cup \{ \infty \},
\]
and define $r(G) = \sup \{ r(u,v) \mid u \not \sim v \text{ for } u,v \in V \} \in \N \cup \{-\infty\}$.
Notice that $r(G) = -\infty$ if and only if all nodes in $G$ are bisimilar.
The following proposition shows that $r(T) \le 2 \log_2(n)+2$ for every unlabelled tree $T$ of size $n$.

\begin{prop}
	\label{prop:fast-game}
	Let $T = (V,E)$ be an unlabelled tree of size $n$. If $u_0,v_0 \in V$ are not bisimilar, then
	$r(u_0,v_0) \le 2 \log_2 (|u_0|+|v_0|)$.
\end{prop}

\begin{proof}
        We prove the lemma by induction on $k = r(u_0,v_0)$. The cases $k = 0$ and $k=1$ are clear since $2 \log_2 (|u_0|+|v_0|) \geq 2$.
        Now assume that $k \geq 2$. Since $u_0 \approx_k v_0$,
	there exist nodes $u_1,u_2,v_1,v_2 \in V$ and numbers $d_1, d_2 \in \N$ such that
	\begin{itemize}
		\item $u_0 \arrow{d_1} u_1 \arrow{d_2} u_2$ and $v_0 \arrow{d_1} v_1 \arrow{d_2} v_2$,
		\item $r(u_1,v_1) = k-1$ and $r(u_2,v_2) = k-2$.
	\end{itemize}
	Let $d = d_1 + d_2$. Since $u_0 \arrow{d} u_2$ there exists $v_0 \arrow{d} v_3$ such that $u_2 \approx_{k-1} v_3$, and
	since $v_0 \arrow{d} v_2$ there exists $u_0 \arrow{d} u_3$ such that $u_3 \approx_{k-1} v_2$. Since $u_2 \not\approx_{k-1} v_2$, we get
	 $u_2 \neq u_3$ and $v_2 \neq v_3$.
	From $r(u_2,v_2) = k-2$ we get $u_3 \approx_{k-1} v_2 \approx_{k-2} u_2 \approx_{k-1} v_3$, which implies $u_3 \approx_{k-2} v_3$.
	On the other hand, $u_3 \approx_{k-1} v_3$ would imply $u_2 \approx_{k-1} v_2$, which contradicts $r(u_2,v_2) = k-2$.
	We therefore have  $r(u_3,v_3) = k-2$.
	Since $u_0 \arrow{d} u_2$, $u_0 \arrow{d} u_3$, and $u_2 \neq u_3$, the nodes $u_2$ and $u_3$ are roots of disjoint subtrees, which implies
	$|u_0| \geq |u_2|+|u_3|$. Analogously, we get $|v_0| \geq |v_2|+|v_3|$, and hence
	%$u_2$, $v_2$, $u_3$ and $v_3$ are roots of pairwise disjoint subtrees, we have
	\[
		|u_0| + |v_0| \ge |u_2| + |v_2| + |u_3| + |v_3|.
	\]
	Observe that either
	\[
		\frac12(|u_0| + |v_0|) \ge |u_2| + |v_2|,
	\]
	or
	\[
		\frac12(|u_0| + |v_0|) \ge |u_3| + |v_3|
	\]
	holds.
	Since the two cases are symmetric we will assume the former.
	The induction hypothesis implies
	\[
		k-2 \le 2 \log_2(|u_2| + |v_2|) \le 2 \log_2(|u_0| + |v_0|) - 2,
	\]
	and hence $k \le 2 \log_2(|u_0|+|v_0|)$.
\end{proof}

\begin{lem} \label{lemma-bound-r}
	Let $G$ be an unlabelled dag with root $u_0$ and $T$ be an unlabelled tree with root $v_0$ of size $n$.
	If $u_0 \sim v_0$, then $r(G \uplus T) \le 2\log_2(n)+2$.
\end{lem}

\begin{proof}
        Assume that $u_0 \sim v_0$.
	Since $u_0$ and $v_0$ are the roots of the dag $G$ and the tree $T$, respectively, 
	every node in $G$ is bisimilar to some node in $T$.
	Furthermore $r$ is invariant under bisimulation, i.e., if $u \sim v$ and $u' \sim v'$, then
	$r(u,u') = r(v,v')$.
	Therefore $r(G \uplus T) = r(T) \le 2\log_2(n)+2$.
\end{proof}
Now we can show the $\AC^1$ upper bound for $\Bisim${\rm (Trees, Graphs)}:

\begin{thm}
	$\Bisim${\rm (Trees, Graphs)} is in $\AC^1$.
\end{thm}

\begin{proof}
        Let $T$ be the input tree and $G$ be the input graph. As explained before, we can assume that $T$ and $G$ are unlabelled.
	Assume we want to test whether node $u_0 \in V(T)$ is bisimilar to node $v_0 \in V(G)$.
	First we restrict $G$ to the subgraph reachable from $v_0$ using a reachability oracle in $\NL \subseteq \AC^1$.
	If $G$ is not acyclic, which can be also tested in $\NL$, then $u_0$ and $v_0$ are not bisimilar.
	If $G$ is a dag, we construct a logspace-uniform $\AC^1$-circuit family which computes the approximants $\approx_k$
	on the disjoint union $G \uplus T$ for all $0 \le k \le 2 \log_2 |T| + 2$.
	A similar construction can be found in \cite{GroheV06}.
	To do so, we first compute for all nodes $u,u'$ the set of path lengths from $u$ to $u'$.
	Note that one can test in $\NL \subseteq \AC^1$ whether there is (no) path of a given length between two given nodes.
	Using this information we can compute $\approx_{k+1}$ by an $\AC^0$-circuit from $\approx_k$.
	
	Finally, we test whether $\approx_k \; = \; \approx_{k+1}$ for $k = \lceil 2 \log_2 |T| + 2 \rceil$.
	If this is not the case, then Lemma~\ref{lemma-bound-r} implies $u_0 \not\sim v_0$ and the circuit outputs false.
	Otherwise, $\approx_k \; = \; \sim$. Hence the circuit outputs true if $u_0 \approx_k v_0$, and false otherwise.
\end{proof}

\section{Further applications}

Before we consider simulation problems, let us present two further application of the techniques from Section~\ref{sec-bisi}.

\subsection{Equality of hereditarily finite sets}

The bisimulation problem on trees in term representation arises in a very natural way.
A {\em hereditarily finite set} is either the empty set $\{\}$ or a set $\{a_1, \dots, a_n\}$
containing finitely many hereditarily finite sets $a_1, \dots, a_n$.
Hereditarily finite sets have a natural string representation over the bracket symbols $\{$ and $\}$.
By counting brackets, one can check in $\TC^0$ whether a string over $\{$ and $\}$ is well-bracketed \cite{BaCo89}.
As before, such a well-bracketed string corresponds to a tree. By induction over the height of trees, one can easily show
that two well-bracketed strings over $\{$ and $\}$ represent the same set if and only if the corresponding trees are bisimilar.
Hence, the tree bisimulation problem for (unlabelled) trees in term representation is equivalent
to the {\em set equality problem}, which asks
whether two such string representations represent the same set.
For example $\{ \{\} \{\} \}$ and $\{\{\}\}$ represent the same set.
From Theorem~\ref{thm:bisi-nc1} we obtain the $\NC^1$ upper bound in the following result. 
The $\NC^1$ lower bound follows from the $\NC^1$-hardness of the bisimulation problem 
for unlabelled trees in term representation (Theorem~\ref{thm-lower-bound-L-NC1}).

\begin{cor}
	The set equality problem is  $\NC^1$-complete with respect to $\AC^0$-reductions.
\end{cor}

\subsection{Model checking on tree-shaped graphs}

Recall that the starting point of our paper was the algorithmic meta theorem of Elberfeld et al.~(Theorem~\ref{thm-elberfeld}) which states that a fixed MSO-property can be model-checked
in logarithmic space on a graph class $\mathcal{C}$ of bounded tree-width. On the other hand, if 
$\mathcal{C}$ has unbounded tree-width, then Theorem~\ref{thm-elberfeld} cannot be applied directly.
In this section, we argue that we can extend Theorem~\ref{thm-elberfeld} to some graph classes
of unbounded tree-width if we restrict to bisimulation-invariant MSO (which coincides with the 
modal $\mu$-calculus).

Consider an MSO-formula $\psi = \psi(x)$ with a single free element variable $x$. We would like to know,
whether for a given graph $G \in \mathcal{C}$ and a node $v \in V(G)$, $\psi$ is satisfied in $v$; $(G,v) \models \psi$ for short.
We say that $\psi$ is {\em bisimulation-invariant} if for all graphs $G$ and all nodes $u,v \in V(G)$
we have: if $u$ and $v$ are bisimilar and $(G,u) \models \psi$ holds, then $(G,v) \models \psi$.
It is known that an MSO-formula $\psi$ is bisimulation-invariant if and only if $\psi$ is equivalent to a formula
in the {\em modal $\mu$-calculus} \cite{JaninW96}.
Assume now that there exist a constant $k$ and a logspace transducer, that maps a graph
$G \in \mathcal{C}$ to a pair $(H,R)$, where $H$ is a graph of tree-width at most $k$ and $R \subseteq V(G) \times V(H)$
is a bisimulation $R$ such that for all $v \in V(G)$ there exists $(v,w) \in R$.
Then we can check in logspace whether $(G,v) \models \psi$ holds, by computing
a bisimilar node $w \in V(H)$ with $(v,w) \in R$ and checking $(H,w) \models \psi$ in logspace
using Theorem~\ref{thm-elberfeld}.

Our second proof of Theorem~\ref{thm:eval-logspace} uses the above strategy. Whether 
a certain gate of a Boolean circuit (viewed as a node-labelled graph) evaluates to true is a bisimulation invariant property.
The partial tree unfolding constructed in the second proof of Theorem~\ref{thm:eval-logspace} has the properties
of the graph $H$ from the previous paragraph. 

The following definition presents a generalization of tree decompositions
that allows to compute a pair $(H,R)$ with $H$ of bounded tree-width and $R \subseteq V(G) \times V(H)$ a bisimulation.

%The idea of partially unfolding a tree-shaped graph can be extended from circuit evaluation problems to testing various graph properties
%which are preserved under the unfolding. For this, we introduce a generalization of tree decompositions.

\begin{defi}
	Let $G = (V,E)$ be an edge-labelled graph. A {\em tree-shaped decomposition} for $G$ is a node-labelled tree $(T,\beta)$
	where $\beta\colon V(T) \to 2^V$ satisfies the following conditions:
	\begin{itemize}
		\item For all $v \in V$ there exists $t \in V(T)$ such that $v \in \beta(t)$.
		\item For all $u \arrow{a} v$ there exists $t \in V(T)$ such that $u,v \in \beta(t)$.
		\item Let $u \arrow{a} v$ and $t_1 \in V(T)$ such that $u \in \beta(t_1)$.
		Then there exists $t_3 \in V(T)$ such that $u,v \in \beta(t_3)$ and for all nodes $t_2 \in V(T)$
		between $t_1$ and $t_3$, we have $u \in \beta(t_2)$.
	\end{itemize}
	The {\em width} of $(T,\beta)$ is $\max_{t \in V(T)} |\beta(t)|$.
\end{defi}
%The idea is that a tree-shaped decomposition of a graph $G$ is a tree decomposition of a partial unfolding of $G$,
%for example as the partial unfolding constructed in the proof of Theorem~\ref{thm:eval-logspace}.
Note that every tree decomposition is also a tree-shaped decomposition but not vice versa
since we do not demand that $\{ t \in V(T) \mid v \in \beta(t) \}$ is connected.
The third property of tree-shaped decompositions ensures that directed edges in $G$
can be ``traced'' by paths in $T$.

\begin{prop}
	\label{prop:general-partial-unfolding}
	For every constant $k \in \N$, there exists a logspace transducer which, 
	given a graph $G$ and a tree-shaped decomposition $(T,\beta)$ for $G$ of width at most $k$,
	outputs a graph $H$ of tree-width at most $k-1$ and a bisimulation $R \subseteq V(G) \times V(H)$ such that
	for all $v \in V(G)$ there exists $(v,w) \in R$.
\end{prop}

\begin{proof}
	Consider the set of pairs
	\[
		P = \{ (v,t) \in V(G) \times V(T) \mid v \in \beta(t) \} .
	\]
	We say that two pairs $(v,t_1), (w,t_3) \in P$ are {\em related} if $v = w$ and for all nodes $t_2 \in V(T)$ on the path
	between $t_1$ and $t_3$ we have $v \in \beta(t_2)$, which can be tested in logspace.
	This relation is an equivalence relation on $P$ and we denote with $[v,t]$ the equivalence class that contains $(v,t)$.
	We take the set of equivalence classes $\{ [v,t] \mid (v,t) \in P \}$
	as the node set of the graph $H$. For all $t \in V(T)$ and $v,w \in \beta(t)$ with $v \arrow{a} w$
	we add the edge $[v,t] \arrow{a} [w,t]$ to $H$. Clearly, $H$ is computable in logspace from $G$ and $(T,\beta)$.
	
	First, we claim that $(T,\gamma)$ is a tree decomposition for $H$ of width at most $k-1$ where
	$\gamma(t) = \{ [v,t] \mid v \in \beta(t) \}$ for all $t \in V(T)$.
	Clearly, every node and every edge of $H$ is contained in some bag $\gamma(t)$. Moreover, every bag $\gamma(t)$ has size at most $k$.
	We need to show that for each node $[v,t] \in V(H)$ the set
	\[
		\{ t' \in V(T) \mid [v,t] \in \gamma(t') \}
	\]
	is connected.
	Let $t_1, t_3 \in V(T)$ such that $[v,t] \in \gamma(t_1) \cap \gamma(t_3)$.
	By definition of $\gamma$ the pairs $(v,t)$, $(v,t_1)$ and $(v,t_3)$ must be related.
	Hence for all $t_2 \in V(T)$ on the unique path between $t_1$ and $t_3$ we have $v \in \beta(t_2)$,
	which implies $[v,t_2] \in \gamma(t_2)$.
	Since $(v,t_2)$ and $(v,t)$ are also related, we conclude $[v,t] \in \gamma(t_2)$.
	
	Now we show that $R = \{ (v,[v,t]) \mid (v,t) \in P \}$ is a bisimulation between $G$ and $H$.
	Let $(v,t) \in P$, i.e., $v \in \beta(t)$.
	If $v \arrow{a} w$ in $G$, then there exists $t' \in V(T)$ such that $v, w \in \beta(t')$
	and the pairs $(v,t)$ and $(v,t')$ are related. Hence we find the matching edge
	$[v,t] = [v,t'] \arrow{a} [w,t']$ in $H$.
	Conversely, if $[v,t] \arrow{a} [w,t]$ is an edge in $H$, then we know by definition that $v \arrow{a} w$
	is an edge in $G$.
\end{proof}
Proposition~\ref{prop:general-partial-unfolding} and Theorem~\ref{thm-elberfeld} together yield:

\begin{thm} \label{thm-mu}
	Let $k \in \N$ be a constant and $\varphi$ be a $\mu$-calculus formula.
	There exists a logspace algorithm which, 
	given a graph $G$, a node $v \in V(G)$ and a tree-shaped decomposition $(T,\beta)$ for $G$ of width at most $k$,
	tests whether $(G,v) \models \varphi$.
\end{thm}
Note that in Theorem~\ref{thm-mu} we need the tree-shaped decomposition $(T,\beta)$ as part of the input. We do not
know whether for every constant $k \in \N$ there exists a logspace transducer that computes from a given graph $G$
a tree-shaped decomposition $(T,\beta)$ for $G$ of width at most $k$ and polynomial size, or outputs $0$ if such a decomposition does not exist.
For tree decompositions, such a logspace transducer indeed exists \cite{ElberfeldJT10}.

\section{Complexity of simulation} \label{simulation}

In this section, we consider simulation problems on restricted classes of finite graphs.
As for bisimulation, it is known that simulation on finite graphs is $\P$-complete.

The upper bounds from Theorem~\ref{thm-bisi-logspace} and~\ref{thm:bisi-nc1} for $\Bisim$(Trees)
carry over to $\Sim$(Trees). The proofs are in fact much easier, since 
the simulation problem for trees reduces to the evaluation of the Boolean circuit obtained from \eqref{bisi-formula}
by removing the second conjunction over all edges $v \arrow{a} v'$; in fact, this circuit is a tree. 
Also the lower bound proofs for Theorem~\ref{thm-lower-bound-L-NC1} hold for the simulation problem: Note that in Figure~\ref{fig:jenner-logspace}(A) the trees
$T_1$ and $T_2$ are isomorphic and hence $T_1$ is simulated by $T_2$. On the other hand, in Figure~\ref{fig:jenner-logspace}(B),
$T_1$ is not simulated by $T_2$. Similarly, for the AND-gadget $T_\wedge(G_1,G_2,H_1,H_2)$ and OR-gadget $T_\vee(G_1,G_2,H_1,H_2)$
(see Figure~\ref{fig-lower-NC1}) from the proof of Theorem~\ref{thm-lower-bound-L-NC1} we have
\begin{itemize}
		\item $G_1 \sqsubseteq H_1$ and $G_2 \sqsubseteq H_2 \iff T_\wedge(G_1,G_2,H_1,H_2) \sqsubseteq T_\wedge(H_1,H_2,G_1,G_2)$, and
		\item $G_1 \sqsubseteq H_1$ or $G_2 \sqsubseteq H_2 \iff T_\vee(G_1,G_2,H_1,H_2) \sqsubseteq T_\vee(G_1,H_2,H_1,G_2)$.
\end{itemize}
Hence we get:

\begin{thm}
The problem {\rm $\Sim$(Trees)} is $\L$-complete (resp., $\NC^1$-complete) if trees are given in pointer 
representation (resp., term representation).
\end{thm}

On the other hand, the construction from the proof of Lemma~\ref{lemma-unlabeling-in-AC0} (elimination of edge labels)
does not carry over to the simulation problem. Indeed, there is a difference in the complexities of the simulation problems 
for edge-labelled trees and unlabelled trees, if trees are given in term representation.
The easy observation is that in unlabelled graphs, $u$ is simulated by $v$ if and only if the longest path starting in $u$
is at most as long as the longest path starting in $v$ (the path lengths can be infinite).

\begin{thm} \label{thm-lower-sim}
The following holds:
\begin{itemize}
\item The problem {\rm $\Sim$(Un}\-{\rm labelled-Graphs)} is $\NL$-complete.
\item The problem {\rm $\Sim$(Unlabelled-Trees)} is $\L$-complete if the trees are given in pointer 
representation.
\item The problem {\rm $\Sim$(Unlabelled-Trees)} is $\TC^0$-complete if the trees are given in term representation.
\end{itemize}
\end{thm}

\begin{proof}
	The upper bounds are easy to show.
	In general graphs we can test in $\NL$ whether from a given node there exists a path of a given length $d \in \N \cup \{ \infty \}$.
	Since $\NL$ is closed under complement, we can also check whether from a given node there does not exist a path of a given length $d \in \N \cup \{ \infty \}$.
	To decide simulation for unlabelled trees we only need to compare the heights,
	which can be done in logspace if the trees are given in pointer representation,
	and in $\TC^0$ if the trees are given in term representation.
	
	For the $\NL$-hardness on general graphs we reduce from the $\NL$-complete problem whether a cycle is reachable from a given node $u$ in a given graph $G$.
	The latter holds if and only if $u$ is simulated by the root of a path of length $|V(G)|$.
	
	For the $\L$-hardness on trees we again reduce from the reachability problem on directed paths \cite{Ete97}.
	Assume that $G$ is a path (which we view as a tree) with $n$ nodes and we want to test whether there exists a path from $u \in V(G)$
	to $v \in V(G)$. Note that such a path exists if and only if the longest path starting in $v$
        is at most as long as the longest path starting in $u$. By the observation before Theorem~\ref{thm-lower-sim} this is equivalent
        to the fact that $v$ is simulated by $u$.
	%We attach a path of length $n$ to $v_j$ and test whether $v_j$ is simulated by $v_i$.

	For the $\TC^0$-hardness we reduce from the $\TC^0$-complete majority problem:
	Given a string $x = a_1 \cdots a_n \in \{ \auf, \zu \}^*$, is the number of opening brackets in $x$
	at least the number of closing brackets in $x$?
	Consider the path $T_1 = \auf^{2n+2} \zu^{2n+2}$ of length $2n+1$ and the tree
	\[
		T_2 = \auf^{n+1} x \; \auf^{n+1} \zu^{n+1} \; \overline{x} \zu^{n+1}
	\]
	where $\overline{x} = \overline{a_n} \cdots \overline{a_1}$ and $\overline{\phantom{a}}$ swaps $\auf$ and $\zu$.
	One can verify that $T_2$ is indeed a valid term representation and that both terms can be computed in $\AC^0$ from the string $x$.
	Notice that the node in $T_2$ that is represented by the innermost bracket pair $\auf \zu$ has maximal depth.
	Now one can show that the root of $T_1$ is simulated by the root of $T_2$ if and only if
	the height of $T_2$ is at least $2n+1$ if and only if
	the number of $\auf$'s in $x$ is at least the number of $\zu$'s. 
\end{proof}
In the rest of Section~\ref{simulation} we complete the complexity picture for the simulation problem on finite graphs with edge labels.
In Section~\ref{simulation-bounded-path-width} we show that $\P$-completeness for the simulation problem already holds for 
graphs of bounded path-width.
In Section~\ref{sec-sim-tree-dag} we show that simulation between a tree and an arbitrary graph,
as well as simulation between an arbitrary graph and a tree is $\LogCFL$-complete.

\subsection{Simulation on graphs of bounded path-width} \label{simulation-bounded-path-width}

In this section we show that the simulation problem is $\P$-complete on
graphs of bounded path-width, and hence also on graphs of bounded tree-width.
It remains open whether the bisimulation problem for graphs of bounded tree-width belongs to $\NC$
or remains $\P$-complete.
For integers $i,j$ we use the abbreviation $[i,j] = \{ k \in \N \mid i \leq k \leq j \}$.
Let PW$_{\le k}$ denote the class of all graphs with path-width at most $k$.

\begin{thm}
	There is a number $k$ such that $\Sim${\rm (PW$_{\le k}$)} \P-complete.
\end{thm}

\begin{proof}
	Fix a \P-complete language $L \subseteq \{0,1\}^*$ and a deterministic polynomial time bounded Turing machine 
	$M = (Q,\Gamma,\{0,1\}, q_0, q_f, \delta)$
	that accepts $L$. Here $Q$ is the set of states, $\Gamma \supseteq \{0,1,\Box\}$
	is the tape alphabet ($\Box$ is the blank symbol), $q_0$ is the initial state, $q_f$ is the final state, and 
	$\delta \colon Q \times \Gamma \to Q \times \Gamma \times \{\to,\gets\}$
	is the transition function, where $\to$ and $\gets$ indicate the head direction. 
	The machine has a single two-way infinite tape, whose cells are indexed with integers. Initially, the input $x$ is written in 
	cells $0, \ldots, |x|-1$ and the tape head scans cell $0$.
	We can assume that there is a polynomial $p(n)$ such that for every input $x \in \{0,1\}^*$ we have:
	$x \in L$ if and only if after $p(|x|)$ many transitions the machine is in state $q_f$, cell $0$ contains $\Box$, and 
	the tape head scans cell $0$. 
	
	We can view configurations of $M$ as words from $\Gamma^* (Q \times \Gamma) \Gamma^*$.
	Let $\Omega = \Gamma \cup (Q \times \Gamma)$. 
	We define a partial mapping $\Delta \colon \Omega^3 \to \Omega$ which
	computes from the three symbols at positions $i-1$, $i$, $i+1$ in a configuration
	the symbol at position $i$ in the successor configuration.
	Formally, for all $a,a',b,c \in \Gamma$, $p,q \in Q$ we define
	\begin{itemize}
		\item $\Delta(a,b,c)=b$ 
		\item $\Delta(b,c,(q,a)) = (p,c)$ if $\delta(q,a) = (p,a',\gets)$
		\item $\Delta(b,c,(q,a)) = c$ if $\delta(q,a) = (p,a',\to)$ 
		\item $\Delta(b,(q,a),c) = a'$  if $\delta(q,a) = (p,a',d)$ for some $d \in \{\to,\gets\}$ 
		\item  $\Delta((q,a),b,c) = b$ if $\delta(q,a) = (p,a',\gets)$ 
		\item $\Delta((q,a),b,c) = (p,b)$ if $\delta(q,a) = (p,a',\to)$ 
	\end{itemize}
	In all other cases, $\Delta$ is undefined.
	
	Let us fix an input $x = a_0 a_1 \cdots a_{n-1}$ of length $n>0$  for the machine $M$ and let $N = p(n)+1$.
	Then there exists a unique computation of $M$ on input $x$. We denote with $C$ the
	corresponding computation table. Formally, it is a mapping
	$C \colon [-N,N] \times [0,N-1] \to \Omega$, where $C(i,t)$ is the symbol at cell $i$ in the $t$-th configuration.
	It can be defined by the following properties:
	\begin{itemize}
		\item $C(0,0) = (a_0,q_0)$, $C(i,0)=a_i$ for $i \in [1,n-1]$, $C(i,0) = \Box$ for $i \in [-N,N]  \setminus [0,n-1]$,
		\item $C(-N,t) = C(N,t) = \Box$ for all $t \in [0,N-1]$
		\item $C(i,t) = \Delta(C(i-1,t-1), C(i,t-1), C(i+1,t-1))$ for all $t \in [1,N-1]$, $i \in [-N+1,N-1]$.
	\end{itemize}
	Let us fix the set of edge labels $A = \{ -1, 0, 1, \alpha \} \uplus \Omega$. We define two edge-labelled graphs  $P$ 
	(for position) and $T$ (for time) with
	edge labels from $A$ and the node sets 
	$$
	V(P) =   [-N,N] \times  \{0,1\},  \quad
	V(T) =   [0,N-1]  \;\cup\; [0,N-1] \times \Omega \;\cup\;  [0,N-1] \times \Omega^3 .
	$$
	For better readability, we write edges of $P$ (resp., $T$) 
	as $x \arrow{a}_P y$ (resp., $x \arrow{a}_T y$). Then, $P$ and $T$ contain the following edges:
	\begin{align*}
		& (i,0)   \arrow{\alpha}_P   (i,1)  \text{ for all }  i \in [-N,N]  \\
		& (i,1)   \arrow{\delta}_P  (i+\delta,0) \text{ for all }  i \in [-N,N], \delta \in \{-1,0,1\} \text{ with }  i+\delta \in [-N,N] \\
		& (i,0)   \xrightarrow{C(i,0)}_P  (i,0)  \text{ for all }  i \in [-N,N] \\
		& (t,a)   \arrow{\alpha}_T    (t-1, b,c,d) \text{ for all } a \in \Omega, t \in [1,N-1], (b,c,d) \in \Delta^{-1}(a) \\
		& (t,a_{-1},a_0,a_1)   \arrow{\delta}_T    (t, a_\delta) \text{ for all } t \in [0,N-2], a_{-1},a_0,a_1 \in \Omega, \delta \in \{-1,0,1\} \\
		& (t,a)   \arrow{b}_T  t  \text{ for all }  a,b \in \Omega, t \in [1,N-1] \\
		& (0,a)   \arrow{b}_T  0 \text{ for all }  a \in \Omega, b \in \{a,\alpha\}  \\
		&  t \arrow{a}_T  t \text{ for all }  t \in [0,N-1], a \in A 
	\end{align*}
	An example of the construction is shown in Figure~\ref{fig:ptime-computation}, where we assume $N=3$ for simplicity.
	It is easy to see that both $P$ and $T$ (and hence also the disjoint union of $P$ and $T$) have bounded path-width. 
	More precisely, $P$ has path-width 3 (the bags are the sets $\{(i,0), (i,1), (i+1,0), (i+1,1)\}$ for $-N \leq i \leq N-1$), whereas
	the path-width of $T$ is bounded by $|\Omega|^3 + |\Omega|$ (the bags are the set
	$\{t, (t,a), (t-1,b,c,d) \mid a,b,c,d \in \Omega \}$ for $1 \leq t \leq N-1$ and $\{(t,a), (t,b,c,d) \mid a,b,c,d \in \Omega \}$ 
	for $0 \leq t \leq N-2$). Recall that $|\Omega|$ is a fixed constant since the machine $M$ is fixed.

	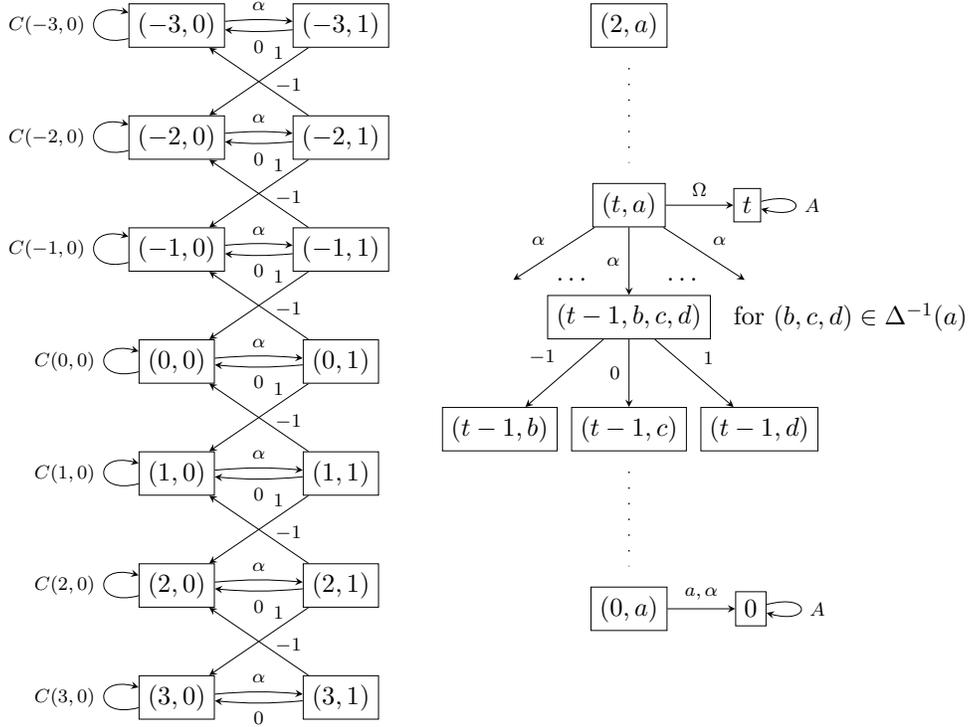
\begin{figure}[t]
		\centering
		
		\tikzstyle{lts} = [->, >=stealth]
		\tikzstyle{state} = [draw]
		\tikzstyle{loopl} = [loop left, distance = .7cm]
		\tikzstyle{loopr} = [loop right, distance = .7cm]
		
		\scalebox{.9}{
			\begin{tikzpicture}[lts]
			
			\node [state] (10) {$(-3,0)$};
			\node [state,below = 1cm of 10] (20) {$(-2,0)$};
			\node [state,below = 1cm of 20] (30) {$(-1,0)$};
			\node [state,below = 1cm of 30] (40) {$(0,0)$};
			\node [state,below = 1cm of 40] (50) {$(1,0)$};
			\node [state,below = 1cm of 50] (60) {$(2,0)$};
			\node [state,below = 1cm of 60] (70) {$(3,0)$};
			
			\node [state,right = 1cm of 10] (11) {$(-3,1)$};
			\node [state,below = 1cm of 11] (21) {$(-2,1)$};
			\node [state,below = 1cm of 21] (31) {$(-1,1)$};
			\node [state,below = 1cm of 31] (41) {$(0,1)$};
			\node [state,below = 1cm of 41] (51) {$(1,1)$};
			\node [state,below = 1cm of 51] (61) {$(2,1)$};
			\node [state,below = 1cm of 61] (71) {$(3,1)$};
			
			\draw[bend left = 5] (10) edge node [above] {\scriptsize $\alpha$} (11);
			\draw[bend left = 5] (20) edge node [above] {\scriptsize $\alpha$} (21);
			\draw[bend left = 5] (30) edge node [above] {\scriptsize $\alpha$} (31);
			\draw[bend left = 5] (40) edge node [above] {\scriptsize $\alpha$} (41);
			\draw[bend left = 5] (50) edge node [above] {\scriptsize $\alpha$} (51);
			\draw[bend left = 5] (60) edge node [above] {\scriptsize $\alpha$} (61);
			\draw[bend left = 5] (70) edge node [above] {\scriptsize $\alpha$} (71);
			
			\draw[bend left = 5] (11) edge node [below] {\scriptsize $0$} (10);
			\draw[bend left = 5] (21) edge node [below] {\scriptsize $0$} (20);
			\draw[bend left = 5] (31) edge node [below] {\scriptsize $0$} (30);
			\draw[bend left = 5] (41) edge node [below] {\scriptsize $0$} (40);
			\draw[bend left = 5] (51) edge node [below] {\scriptsize $0$} (50);
			\draw[bend left = 5] (61) edge node [below] {\scriptsize $0$} (60);
			\draw[bend left = 5] (71) edge node [below] {\scriptsize $0$} (70);

			\draw (11) edge node [above, pos=0.3] {\scriptsize $1$} (20);
			\draw (21) edge node [above, pos=0.3] {\scriptsize $1$} (30);
			\draw (31) edge node [above, pos=0.3] {\scriptsize $1$} (40);
			\draw (41) edge node [above, pos=0.3] {\scriptsize $1$} (50);
			\draw (51) edge node [above, pos=0.3] {\scriptsize $1$} (60);
			\draw (61) edge node [above, pos=0.3] {\scriptsize $1$} (70);
			
			\draw (21) edge node [above, pos=0.2] {\scriptsize $-1$} (10);
			\draw (31) edge node [above, pos=0.2] {\scriptsize $-1$} (20);
			\draw (41) edge node [above, pos=0.2] {\scriptsize $-1$} (30);
			\draw (51) edge node [above, pos=0.2] {\scriptsize $-1$} (40);
			\draw (61) edge node [above, pos=0.2] {\scriptsize $-1$} (50);
			\draw (71) edge node [above, pos=0.2] {\scriptsize $-1$} (60);
			
			\draw[loopl] (10) edge node {\scriptsize $C(-3,0)$} (10);
			\draw[loopl] (20) edge node {\scriptsize $C(-2,0)$} (20);
			\draw[loopl] (30) edge node {\scriptsize $C(-1,0)$} (30);
			\draw[loopl] (40) edge node {\scriptsize $C(0,0)$} (40);
			\draw[loopl] (50) edge node {\scriptsize $C(1,0)$} (50);
			\draw[loopl] (60) edge node {\scriptsize $C(2,0)$} (60);
			\draw[loopl] (70) edge node {\scriptsize $C(3,0)$} (70);

			\node[state, right = 5.4cm of 10] (first) {$(2,a)$};
			
			\node[state, below = 2cm of first] (s0) {$(t,a)$};
			\node[state, below = 1cm of s0] (s1) {$(t-1,b,c,d)$};
			\node[above left = 0cm and .5cm of s1] (s2) {};
			\node[above right = 0cm and .5cm of s1] (s3) {};
			\node[state, below = 1cm of s1] (s4) {$(t-1,c)$};
			\node[state, left = .2cm of s4] (s5) {$(t-1,b)$};
			\node[state, right = .2cm of s4] (s6) {$(t-1,d)$};
			
			\node[right = .2cm of s1] {for $(b,c,d) \in \Delta^{-1}(a)$};
			\node[above left = -.2cm and .5cm of s3] {$\dots$};
			\node[above right = -.2cm and .5cm of s2] {$\dots$};
			
			\draw[etc] (first) edge (s0);
			
			\draw (s0) edge node [left] {\scriptsize $\alpha$} (s1);
			\draw (s0) edge node [above left] {\scriptsize $\alpha$} (s2);
			\draw (s0) edge node [above right] {\scriptsize $\alpha$} (s3);
			\draw (s1) edge node [left] {\scriptsize $0$} (s4);
			\draw (s1) edge node [above left] {\scriptsize $-1$} (s5);
			\draw (s1) edge node [above right] {\scriptsize $1$} (s6);
			
			\node[state, right = 1cm of s0] (s00) {$t$};
			
			\draw (s0) edge node [above] {\scriptsize $\Omega$} (s00);
			\draw[loopr] (s00) edge node {\scriptsize $A$} (s00);
			
			\node[state, below = 2cm of s4] (last) {$(0,a)$};
			\draw[etc] (s4) edge (last);
			
			\node[state, right = 1cm of last] (last0) {$0$};
			
			\draw (last) edge node [above] {\scriptsize $a,\alpha$} (last0);
			\draw[loopr] (last0) edge node {\scriptsize $A$} (last0);
			
			\end{tikzpicture}
		}
		
		\caption{The edge-labelled graphs $P$ and $T$, where we assume $N=3$ for simplicity.}
		\label{fig:ptime-computation}
	\end{figure}

	The following claim proves the theorem, since $M$ accepts $x$ if and only if $C(0,N-1) = (q_f, \Box)$ by our assumptions on $M$.
	
	\subparagraph*{Claim.} For all $t \in [0,N-1]$, $i \in [-N,N]$ and
	$a \in \Omega$ such that $i+t \leq N$ and $i-t \geq -N$ we have:
	$C(i,t) = a$ if and only if $(i,0) \sqsubseteq (t,a)$.  
	
	We prove the claim by induction on $t$. First, note that for every  $v \in V(P)$  and every  $t \in [0,N-1] \subseteq V(T)$ 
	we have $v \sqsubseteq t$, since at $t$ we can loop with every $a \in A$.
	For $t = 0$ note that indeed $(i,0) \sqsubseteq (0, C(i,0))$: The only outgoing edges 
	for $(i,0)$ are labelled with $C(i,0)$ and $\alpha$. From $(0,C(i,0))$ these labels lead to node $0$, which
	simulates every node of $P$. On the other hand, if $a \neq C(i,0)$, then $(i,0)$ has a 
	$C(i,0)$-labelled outgoing edge, whereas $(0,a)$ has no such outgoing edge.
	This implies $(i,0) \not\sqsubseteq (0,a)$.
	
	Now assume that $t \in [1,N-1]$, $i \in [-N,N]$, $i+t \leq N$, and $i-t \geq -N$
	and that the claim holds for $t-1$. First assume that $C(i,t) = a$. Since $t \geq 1$,
	we have $i+1 \leq N$ and $i-1 \geq -N$, i.e., $i \in [-N+1,N-1]$.
	We have to show that $(i,0) \sqsubseteq (t,a)$. The edge
	$(i,0)   \xrightarrow{C(i,0)}_P  (i,0)$ can be simulated by the edge $(t,a)  \xrightarrow{C(i,0)}_T t$ (recall that $(i,0) \sqsubseteq t$). 
	Now consider the other possible edge $(i,0)   \arrow{\alpha}_P  (i,1)$. Since $C(i,t) = a$,
	there must exist $(b,c,d) \in \Delta^{-1}(a)$ such that $b = C(i-1,t-1)$, $c = C(i,t-1)$, and $d = C(i+1,t-1)$.
	Also note that $i+\delta+(t-1) \leq N$ and $i+\delta-(t-1) \geq  -N$ for all $\delta \in \{-1,0,1\}$.
	Hence, by induction $(i-1,0) \sqsubseteq (t-1,b)$, $(i,0) \sqsubseteq (t-1,c)$, and  $(i+1,0) \sqsubseteq (t-1,d)$.
	But this implies that $(i,1) \sqsubseteq (t-1, b,c,d)$.
	Hence, we can choose the edge $(t,a)   \arrow{\alpha}_T    (t-1, b,c,d)$ in order to simulate the edge 
	$(i,0)   \arrow{\alpha}_P  (i,1)$. 
	
	Finally, assume that $C(i,t) \neq a$. We have to show that $(i,0) \not\sqsubseteq (t,a)$.
	Let us choose the edge $(i,0)   \arrow{\alpha}_P  (i,1)$. We have to show that for every
	$(b,c,d) \in \Delta^{-1}(a)$, $(i,1) \not\sqsubseteq (t,b,c,d)$. Let us fix a triple 
	$(b,c,d) \in \Delta^{-1}(a)$. Since $C(i,t) \neq a$, one of the following three statements holds:
	$C(i-1,t-1) \neq b$, $C(i,t-1) \neq c$, $C(i+1,t-1) \neq d$. Hence, by induction, 
	$(i-1,0) \not\sqsubseteq (t-1,b)$ or $(i,0)  \not\sqsubseteq (t-1,c)$ or  $(i+1,0) \not\sqsubseteq (t-1,d)$.
	This implies that, indeed, $(i,1) \not\sqsubseteq (t,b,c,d)$.
\end{proof}
It seems to be difficult to modify the above proof so that it shows $\P$-hardness for bisimulation on
graphs of bounded path-width or bounded tree-width.  One might try to restrict the choices
of the players in the bisimulation game (see e.g.~\cite{AcIng07}) so that they are forced to play as in the simulation game.
There is a technique to achieve this (defenders forcing) but the problem is that it yields grid-like
graph structures and hence graphs of unbounded tree-width.

\subsection{Simulation between  trees and general graphs} \label{sec-sim-tree-dag}

In this section we show that the problems $\Sim$(Trees, Graphs)  and $\Sim$(Graphs, Trees) are both
\LogCFL-complete. Since simulation is not symmetric, it is not obvious that these two problems are equivalent.

\begin{thm}
	\label{thm:trees-graphs}
	\Sim{\rm (Trees, Graphs)} is \LogCFL-complete.
\end{thm}

\begin{proof}
Again we can restrict ourselves to the case that the input vertices are the roots of the graph and the tree, respectively.
We start with showing that the simulation problem in the theorem belongs to $\LogCFL$. For this we outline a
straightforward alternating
logspace algorithm with polynomial proof tree size. 
Given a tree $T$ and a graph $G$, the algorithm stores a pair $(u,w) \in V(T) \times V(G)$
and continues as follows:
\begin{itemize}
\item Branch universally to all edges $u \arrow{a} u'$ in $T$.
\item Branch existentially to all edges $w \arrow{a} w'$ in $G$.
\item Continue with the pair $(u',w')$.
\end{itemize}
The algorithm starts with the pair $(r,v)$. It is obvious that there exists a proof tree if and only if $r \sqsubseteq v$.
Moreover, a proof tree results from $T$ by inserting on each edge a node of $G$ with a unique child.

Let us now prove the lower bound.
Fix a $\LogCFL$-complete problem $L$ and let $w$ be an input word of length $n$. Let $(C_n)_{n \geq 0}$ be a logspace-uniform
 $\SAC^1$ circuit family for $L$.
We first construct in logspace the $\SAC^1$-circuit $C_n$ for input length $n$. 
We set the input gates to the bits in $w$; let $C = (G_0,\beta)$ be the resulting  Boolean circuit. 
Recall from Section~\ref{sec-circuits} that $G_0$ is a directed graph $G_0 = (V_0,E_0)$ and that
edges from $E_0$ are directed towards the input gates.

Our goal is to reduce the question whether $C$ evaluates to $1$ to a simulation problem. By \cite[Lemma~4.6]{GottlobLS01}
we can assume that $C$ is layered and alternating, i.e., the set of gates is partitioned into $m \in O(\log n)$
many layers $L_1, L_2, \ldots,  L_m$, where $m>1$ is odd, $L_1$ contains the input gates, all wires go from layer $L_{i}$ to
layer $L_{i-1}$ for some $2 \leq i \leq m$, and $L_m = \{ g_0 \}$ contains the output gate of the circuit.
Moreover,  layers $2i$  ($2i+1$, resp.)  contain only $\wedge$-gates ($\vee$-gates, resp.) 
for $1 \leq i \leq \lfloor m/2 \rfloor$. Since $m$ is odd, the output gate $g_0$ is a $\vee$-gate.
Recall that every $\wedge$-gate $u$ has exactly two input nodes that we denote with $u_1$ and $u_2$
(we allow $u_1 = u_2$).

\begin{figure}
\centering
	\begin{tikzpicture}
	
	\tikzset{gate/.append style={fill=white}}
	
	\node[gate] (05) {\scriptsize $\vee$};
	
	\node[gate, below left = 20pt and 15pt of 05] (35) {\scriptsize $\wedge$};
	\node[gate, left = 25pt of 35] (15) {\scriptsize $\wedge$};
	\node[gate, below right = 20pt and 15pt of 05] (06) {\scriptsize $\wedge$};
	\node[gate, right = 25pt of 06] (09) {\scriptsize $\wedge$};
	
	\draw[arrow] (05) edge node [above left, inner sep = .5] {\scriptsize $a$} (35);
	\draw[arrow] (05) edge node [above left, inner sep = .5] {\scriptsize $a$} (15);
	\draw[arrow] (05) edge node [above right, inner sep = .5] {\scriptsize $a$} (06);
	\draw[arrow] (05) edge node [above right, inner sep = .5] {\scriptsize $a$} (09);
	
	\node[gate, below left = 20pt and 5pt of 15] (16) {\scriptsize $\vee$};
	\node[gate, right = 25pt of 16] (19) {\scriptsize $\vee$};
	\node[gate, right = 25pt of 19] (36) {\scriptsize $\vee$};
	\node[gate, right = 25pt of 36] (39) {\scriptsize $\vee$};
	\node[gate, right = 25pt of 39] (40) {\scriptsize $\vee$};

	\draw[arrow] (15) edge node [above left, inner sep = .5] {\scriptsize $a$} (16);
	\draw[arrow] (15) edge node [above right, inner sep = .5] {\scriptsize $b$} (19);
	\draw[arrow] (35) edge node [above left, inner sep = .5] {\scriptsize $a$} (19);
	\draw[arrow] (35) edge node [above right, inner sep = .5] {\scriptsize $b$} (36);
	\draw[arrow] (06) edge node [above left, inner sep = .5] {\scriptsize $a$} (36);
	\draw[arrow] (06) edge node [above right, inner sep = .5] {\scriptsize $b$} (39);
	\draw[arrow] (09) edge node [above left, inner sep = .5] {\scriptsize $a$} (39);
	\draw[arrow] (09) edge node [above right, inner sep = .5] {\scriptsize $b$} (40);
	
	\node[gate, below left = 40pt and 15pt of 19] (40) {\scriptsize $\wedge$};
	\node[gate, below right = 40pt and 20pt of 19] (41b) {\scriptsize $\wedge$};
	\node[gate, below right = 40pt and 15pt of 39] (41) {\scriptsize $\wedge$};
	\node[below = 40pt of 19] (42) {};
	\node[below = 40pt of 39] (42b) {};
	
	\node[gate, below left = 20pt and 10pt of 40] (43) {\scriptsize $0$};	
	\node[gate, right = 20pt of 43] (44) {\scriptsize $1$};	
	\node[gate, right = 20pt of 44] (45) {\scriptsize $0$};	
	\node[gate, right = 20pt of 45] (46) {\scriptsize $1$};	
	\node[gate, right = 20pt of 46] (47) {\scriptsize $1$};
	\node[gate, right = 20pt of 47] (48) {\scriptsize $0$};	
	
	\draw[arrow] (40) edge node [above left, inner sep = .5] {\scriptsize $a$} (43);
	\draw[arrow] (41b) edge node [above left, inner sep = .5] {\scriptsize $a$} (45);
	\draw[arrow] (41) edge node [above left, inner sep = .5] {\scriptsize $a$} (47);
	\draw[arrow] (40) edge node [above right, inner sep = .5] {\scriptsize $b$} (44);
	\draw[arrow] (41b) edge node [above right, inner sep = .5] {\scriptsize $b$} (46);
	\draw[arrow] (41) edge node [above right, inner sep = .5] {\scriptsize $b$} (48);
		
	\node[gate, below right = 20pt and 5pt of 45] (w) {\scriptsize $\omega$};
	
	\draw [etc] (19) -- (42);
	\draw [etc] (39) -- (42b);

	\draw[arrow] (44) edge node [below left, inner sep = 1] {\scriptsize $a$} (w);
	\draw[arrow] (46) edge node [above left, inner sep = .5] {\scriptsize $a$} (w);
	\draw[arrow] (47) edge node [below right, inner sep = 1] {\scriptsize $a$} (w);

	\node [right = 90pt of 05] (m) {$L_m$};
	\node [below = 10pt of m] (m1) {$L_{m-1}$};
	\node [below = 10pt of m1] (m2) {$L_{m-2}$};
	\node [below = 35pt of m2] (m3) {$L_2$};
	\node [below = 10pt of m3] (m4) {$L_1$};

	% Tree
	
	\node [gate, inner sep = 3pt, left = 180pt of 05] (t1) {};
	\node [gate, inner sep = 3pt, below = 20pt of t1] (t2) {};
	\node [gate, inner sep = 3pt, below left = 20pt and 20pt of t2] (t3) {};
	\node [gate, inner sep = 3pt, below right = 20pt and 20pt of t2] (t4) {};

	\draw[arrow] (t1) edge node [right] {\scriptsize $a$} (t2);
	\draw[arrow] (t2) edge node [above left] {\scriptsize $a$} (t3);
	\draw[arrow] (t2) edge node [above right] {\scriptsize $b$} (t4);

	\draw[arrow] (t3) edge node [left] {\scriptsize $a$} +(0,-20pt);
	\draw[arrow] (t4) edge node [right] {\scriptsize $a$} +(0,-20pt);
	\draw [etc] (t2) ++ (0,-40pt) edge +(0,-40pt);

	\node [gate, inner sep = 3pt, below left = 45pt and 14pt of t3] (t5) {};
	\node [gate, inner sep = 3pt, right = 22pt of t5] (t6) {};
	\node [gate, inner sep = 3pt, below left = 45pt and 6pt of t4] (t7) {};
	\node [gate, inner sep = 3pt, right = 22pt of t7] (t8) {};

	\node [gate, inner sep = 3pt, below left = 20pt and 1pt of t5] (t11) {};
	\node [gate, inner sep = 3pt, below right = 20pt and 2pt of t5] (t12) {};
	\node [gate, inner sep = 3pt, below left = 20pt and 1pt of t6] (t13) {};
	\node [gate, inner sep = 3pt, below right = 20pt and 2pt of t6] (t14) {};
	\node [gate, inner sep = 3pt, below left = 20pt and 1pt of t7] (t15) {};
	\node [gate, inner sep = 3pt, below right = 20pt and 2pt of t7] (t16) {};
	\node [gate, inner sep = 3pt, below left = 20pt and 1pt of t8] (t17) {};
	\node [gate, inner sep = 3pt, below right = 20pt and 2pt of t8] (t18) {};

	\draw[arrow] (t5) edge node [above left, inner sep = .5] {\scriptsize $a$} (t11);
	\draw[arrow] (t5) edge node [above right, inner sep = .5] {\scriptsize $b$} (t12);
	\draw[arrow] (t6) edge node [above left, inner sep = .5] {\scriptsize $a$} (t13);
	\draw[arrow] (t6) edge node [above right, inner sep = .5] {\scriptsize $b$} (t14);
	\draw[arrow] (t7) edge node [above left, inner sep = .5] {\scriptsize $a$} (t15);
	\draw[arrow] (t7) edge node [above right, inner sep = .5] {\scriptsize $b$} (t16);
	\draw[arrow] (t8) edge node [above left, inner sep = .5] {\scriptsize $a$} (t17);
	\draw[arrow] (t8) edge node [above right, inner sep = .5] {\scriptsize $b$} (t18);

	\node [gate, inner sep = 3pt, below = 18pt of t11] (at11) {};
	\node [gate, inner sep = 3pt, below = 18pt of t12] (at12) {};
	\node [gate, inner sep = 3pt, below = 18pt of t13] (at13) {};
	\node [gate, inner sep = 3pt, below = 18pt of t14] (at14) {};
	\node [gate, inner sep = 3pt, below = 18pt of t15] (at15) {};
	\node [gate, inner sep = 3pt, below = 18pt of t16] (at16) {};
	\node [gate, inner sep = 3pt, below = 18pt of t17] (at17) {};
	\node [gate, inner sep = 3pt, below = 18pt of t18] (at18) {};

	\foreach \x in {1,...,8} {
	\draw[arrow] (t1\x) edge node [left, inner sep = .5] {\scriptsize $a$} (at1\x);
	}
	\end{tikzpicture}
	
	\caption{Reduction from $\SAC^1$-circuits to $\Sim${\rm (Trees, Graphs)}}
	\label{fig:sac-lower}
	
\end{figure}
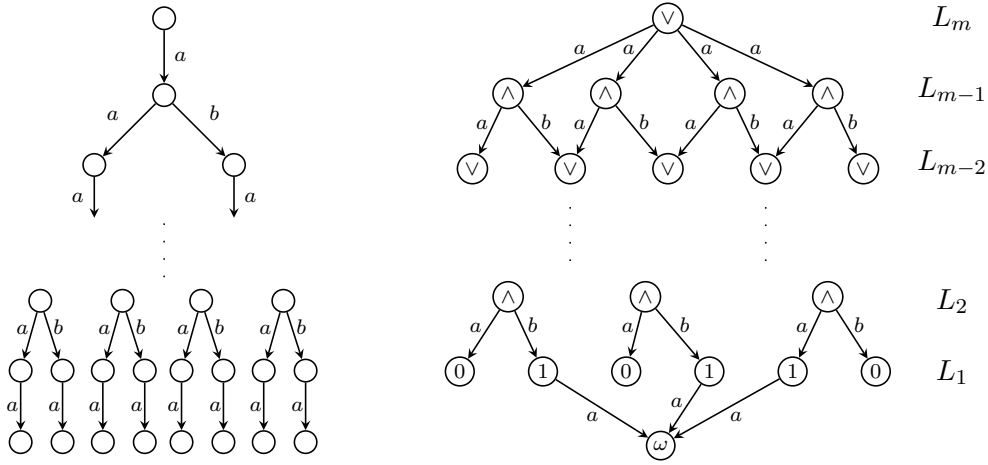

Define the edge-labelled graph $G = (V_0 \uplus \{ \omega \},E)$, where $E$ consists of the following set
of $\{a,b\}$-labelled edges.
\begin{itemize}
\item $v \arrow{a} \omega$ for all $v \in L_1$ with $\beta(v)=1$.
\item $v \arrow{a} v'$ for all $v \in L_i$ with $i > 1$ odd and $(v,v') \in E_0$.
\item $v \arrow{a} v_1$, $v \arrow{b} v_2$ for all $v \in L_i$ with $i > 1$ even, where $v_1, v_2$ are the left and right child of $v$ in $G_0$.
\end{itemize}
The tree $T$ is a full binary tree of height $\lfloor m/2 \rfloor$ with intermediate layers of unary nodes inserted.
More precisely, we define trees $T_i$ ($1 \leq i \leq m$) recursively as follows:
\begin{itemize}
\item If $i=1$ then $T_i$ consists of a root node $r$ with a unique outgoing edge $r \arrow{a} r'$, where
$r'$ is a leaf.
\item If $i > 1$ is odd, then $T_i$ consists of a root node $r$ with a unique outgoing edge $r \arrow{a} r'$, where
$r'$ is the root of a copy of $T_{i-1}$.
\item If $i > 1$ is even, then $T_i$ consists of a root node $r$ with two outgoing edges
$r \arrow{a} r_1$ and $r \arrow{b} r_2$. Each of the nodes $r_1$ and $r_2$ is the root of a copy of 
$T_{i-1}$.
\end{itemize}

The tree $T$ is defined as $T_m$.
Both $T$ and $G$ are displayed in Figure~\ref{fig:sac-lower}.

\subparagraph*{Claim.}  If $r$ denotes the root of $T_i$ and $v \in V(G)$ belongs to layer $L_i$ then $r \sqsubseteq v$ if and only if 
gate $v$ evaluates to $1$ in the circuit $C$.

We prove the claim by induction on $i$. The case $i=1$ is clear: In the tree $T_1$ the root node $r$ has a unique
outgoing edge $r \arrow{a} r'$, where $r'$ has no outgoing edges.
Moreover, layer $L_1$ contains the input gates of the circuit $C$. If $v$ is labelled with $0$, then $v$ has no outgoing edges in $G$,
and hence $r \not\sqsubseteq v$. On the other hand, if $v$ is labelled with $1$, then $v$ has a single outgoing $a$-labelled 
edge to a node without outgoing edges, which implies $r \sqsubseteq v$.

Now assume that $i > 1$. Let us first assume that $i$ is odd. 
Then, the root $r$ of $T_i$ has a single outgoing edge $r \arrow{a} r'$, where 
$r'$ is a root of $T_{i-1}$. Moreover, $v$ is a
$\vee$-gate. If $v$ evaluates to $1$, then $v$ has an outgoing edge $v \arrow{a} v'$,
where $v' \in L_{i-1}$ evaluates to $1$.  By induction, we get $r' \sqsubseteq v'$,
which implies $r \sqsubseteq v$. On the other hand, if $r \sqsubseteq v$, then $v$ must have an 
outgoing edge $(v,a,v')$ such that $r' \sqsubseteq v'$. By
induction, it follows that $v'$ evaluates to $1$. Hence, the $\vee$-gate $v$ evaluates to $1$ too.

Now assume that $i$ is even. Then, the root $r$ of $T_i$ has two outgoing edges $r \arrow{a} r_1$ and $r \arrow{b} r_2$
where $r_i$ is a root of a copy of $T_{i-1}$. 
Moreover, $v$ is a $\wedge$-gate, which has two input gates $v_1$ and $v_2$ (we might have $v_1 = v_2$).
In $G$, $v$ has the two outgoing edges $v \arrow{a} v_1$ and $v \arrow{b} v_2$.
If $v$ evaluates to $1$, then $v_1$ and $v_2$ evaluate to $1$. By induction, we get $r_1 \sqsubseteq v_1$ and 
$r_2 \sqsubseteq v_2$, which implies $r \sqsubseteq v$. 
On the other hand, if $r \sqsubseteq v$, then  we must have $r_1 \sqsubseteq v_1$ and $r_2 \sqsubseteq v_2$.
Hence, by induction, it follows that $v_1$ and $v_2$ both evaluate to $1$. Thus, $v$ evaluates to $1$.
\end{proof}

\begin{thm}
	\Sim{\rm (Graphs, Trees)} is \LogCFL-complete.
\end{thm}

\begin{proof}
	The proof of the upper bound is very similar to the previous proof,
	using an alternating logspace machine with a polynomial proof tree size.
	Since $\LogCFL$ is closed under complement \cite{BCDRT89}, we can consider the problem
	of deciding whether a graph node $u \in V(G)$
	is not simulated by a tree node $w \in V(T)$.
	The algorithm stores the pair $(u,w) \in V(G) \times V(T)$ and continues as follows:
	\begin{itemize}
	\item Branch existentially to all edges $u \arrow{a} u'$ in $G$.
	\item Branch universally to all edges $w \arrow{a} w'$ in $T$.
	\item Continue with the pair $(u',w')$.
	\end{itemize}
	The correctness and the polynomial proof tree size of the algorithm can be argued as above.
	
	For the $\LogCFL$-hardness we can use the same graph $G$ as in the proof of Theorem~\ref{thm:trees-graphs} but define a different tree $T$.
	First we reconsider the tree used in the proof of Theorem~\ref{thm:trees-graphs}, which we call $T_*$.
	Notice that every node in $G$ which is contained in an odd level $L_i$ is simulated by the root of $T_*$.
	Then we define trees $T_i$ for all odd $1 \leq i \leq m$ and trees $T_{i,a}, T_{i,b}$ for all even $1 \le i \le m$
	recursively as follows:
	\begin{itemize}
	\item If $i=1$ then $T_i$ consists of a single node.
	\item If $i > 1$ is odd, then $T_i$ consists of a root node $r$ with two outgoing edges
	$r \arrow{a} r_a$ and $r \arrow{a} r_b$.
	The nodes $r_a$ and $r_b$ are the roots of copies of the trees $T_{i-1,a}$ and $T_{i-1,b}$, respectively.
	\item If $i > 1$ is even, then $T_{i,a}$ consists of a root node $r$ with two outgoing edges
	$r \arrow{a} r_a$ and $r \arrow{b} r_b$.
	The node $r_a$ is the root of a copy of $T_*$, and $r_b$ is the root of a copy of $T_{i-1}$.
	\end{itemize}
	Finally the tree $T$ is defined as $T_m$, which is displayed together with $G$ in Figure~\ref{fig:sac-lower-forcing}.

\begin{figure}
\centering
	\begin{tikzpicture}
	
	\tikzset{gate/.append style={fill=white}}
	
	\node[gate] (05) {\scriptsize $\vee$};
	
	\node[gate, below left = 20pt and 15pt of 05] (35) {\scriptsize $\wedge$};
	\node[gate, left = 25pt of 35] (15) {\scriptsize $\wedge$};
	\node[gate, below right = 20pt and 15pt of 05] (06) {\scriptsize $\wedge$};
	\node[gate, right = 25pt of 06] (09) {\scriptsize $\wedge$};
	
	\draw[arrow] (05) edge node [above left, inner sep = .5] {\scriptsize $a$} (35);
	\draw[arrow] (05) edge node [above left, inner sep = .5] {\scriptsize $a$} (15);
	\draw[arrow] (05) edge node [above right, inner sep = .5] {\scriptsize $a$} (06);
	\draw[arrow] (05) edge node [above right, inner sep = .5] {\scriptsize $a$} (09);
	
	\node[gate, below left = 20pt and 5pt of 15] (16) {\scriptsize $\vee$};
	\node[gate, right = 25pt of 16] (19) {\scriptsize $\vee$};
	\node[gate, right = 25pt of 19] (36) {\scriptsize $\vee$};
	\node[gate, right = 25pt of 36] (39) {\scriptsize $\vee$};
	\node[gate, right = 25pt of 39] (40) {\scriptsize $\vee$};

	\draw[arrow] (15) edge node [above left, inner sep = .5] {\scriptsize $a$} (16);
	\draw[arrow] (15) edge node [above right, inner sep = .5] {\scriptsize $b$} (19);
	\draw[arrow] (35) edge node [above left, inner sep = .5] {\scriptsize $a$} (19);
	\draw[arrow] (35) edge node [above right, inner sep = .5] {\scriptsize $b$} (36);
	\draw[arrow] (06) edge node [above left, inner sep = .5] {\scriptsize $a$} (36);
	\draw[arrow] (06) edge node [above right, inner sep = .5] {\scriptsize $b$} (39);
	\draw[arrow] (09) edge node [above left, inner sep = .5] {\scriptsize $a$} (39);
	\draw[arrow] (09) edge node [above right, inner sep = .5] {\scriptsize $b$} (40);
	
	\node[below = 40pt of 19] (42) {};
	\node[below = 40pt of 39] (42b) {};
	
	\node[gate, below left = 40pt and 0pt of 16] (43) {\scriptsize $0$};	
	\node[gate, right = 20pt of 43] (44) {\scriptsize $1$};	
	\node[gate, right = 20pt of 44] (45) {\scriptsize $0$};	
	\node[gate, right = 20pt of 45] (46) {\scriptsize $1$};	
	\node[gate, right = 20pt of 46] (47) {\scriptsize $1$};
	\node[gate, right = 20pt of 47] (48) {\scriptsize $0$};

	\node[gate, below right = 20pt and 5pt of 45] (w) {\scriptsize $\omega$};
	
	\draw [etc] (19) -- (42);
	\draw [etc] (39) -- (42b);

	\draw[arrow] (44) edge node [below left, inner sep = 1] {\scriptsize $a$} (w);
	\draw[arrow] (46) edge node [above left, inner sep = .5] {\scriptsize $a$} (w);
	\draw[arrow] (47) edge node [below right, inner sep = 1] {\scriptsize $a$} (w);

	\node [right = 90pt of 05] (m) {$L_m$};
	\node [below = 10pt of m] (m1) {$L_{m-1}$};
	\node [below = 10pt of m1] (m2) {$L_{m-2}$};
	\node [below = 35pt of m2] (m3) {$L_1$};

	% Tree

	\node [gate, inner sep = 3pt, right = 180pt of 05] (t2) {};
	\node [gate, inner sep = 3pt, below left = 20pt and 20pt of t2] (t3) {};
	\node [gate, inner sep = 3pt, below right = 20pt and 20pt of t2] (t4) {};

	\draw[arrow] (t2) edge node [above left] {\scriptsize $a$} (t3);
	\draw[arrow] (t2) edge node [above right] {\scriptsize $a$} (t4);

	\node [inner sep = 0, below left = 22pt and 10pt of t3] (t3a) {$T_*$};
	\node [gate, inner sep = 3pt, below right = 22pt and 10pt of t3] (t3b) {};
	\node [gate, inner sep = 3pt, below left = 22pt and 10pt of t4] (t4a) {};
	\node [inner sep = 0, below right = 22pt and 10pt of t4] (t4b) {$T_*$};

	\draw[arrow] (t3) edge node [above left] {\scriptsize $a$} (t3a);
	\draw[arrow] (t3) edge node [above right] {\scriptsize $b$} (t3b);
	\draw[arrow] (t4) edge node [above left] {\scriptsize $a$} (t4a);
	\draw[arrow] (t4) edge node [above right] {\scriptsize $b$} (t4b);

	\draw [etc] (t3b) edge +(0,-40pt);
	\draw [etc] (t4a) edge +(0,-40pt);

	\node [inner sep = 3pt, below left = 47pt and 13pt of t3] (t5) {};
	\node [inner sep = 3pt, right = 22pt of t5] (t6) {};
	\node [inner sep = 3pt, below left = 47pt and 7pt of t4] (t7) {};
	\node [inner sep = 3pt, right = 22pt of t7] (t8) {};

	\node [gate, inner sep = 3pt, below left = 20pt and 1pt of t5] (t11) {};
	\node [gate, inner sep = 3pt, below right = 20pt and 1pt of t5] (t12) {};
	\node [gate, inner sep = 3pt, below left = 20pt and 1pt of t6] (t13) {};
	\node [gate, inner sep = 3pt, below right = 20pt and 1pt of t6] (t14) {};
	\node [gate, inner sep = 3pt, below left = 20pt and 1pt of t7] (t15) {};
	\node [gate, inner sep = 3pt, below right = 20pt and 1pt of t7] (t16) {};
	\node [gate, inner sep = 3pt, below left = 20pt and 1pt of t8] (t17) {};
	\node [gate, inner sep = 3pt, below right = 20pt and 1pt of t8] (t18) {};

	\end{tikzpicture}
	
	\caption{Reduction from $\SAC^1$-circuits to $\Sim${\rm (Graphs, Trees)}}
	\label{fig:sac-lower-forcing}
	
\end{figure}

	\subparagraph*{Claim.}  Assume $v \in V(G)$ belongs to layer $L_i$.
	\begin{itemize}
		\item If $i$ is odd, let $r$ be the root of $T_i$.
		Then $v \not \sqsubseteq r$ if and only if gate $v$ evaluates to $1$ in the circuit $C$.
		\item If $i$ is even, let $r_a,r_b$ be the roots of $T_{i,a}$ and $T_{i,b}$, respectively. 
		Then ($v \not \sqsubseteq r_a$ and $v \not \sqsubseteq r_b$) if and only if gate $v$ evaluates to $1$ in the circuit $C$.
	\end{itemize}
	For $i = 1$ this is clear because $v$ has a successor if and only if it evaluates to $1$.
	Now let $i > 1$.  % and let $r_a, r_b$ be the two successors of $r$ as defined above.
	First assume that  $i$ is odd, i.e., $v$ is a $\vee$-gate. Let $r$ be the root of $T_i$.
	Thus, $r$ has two outgoing edges $r \arrow{a} r_a$ and $r \arrow{a} r_b$, where
	$r_a$ and $r_b$ are the roots of copies of the trees $T_{i-1,a}$ and $T_{i-1,b}$, respectively.
	The node $v$ evaluates to 1 in $C$ if and only if $v$ has some $a$-successor $w$ in $G$
	which evaluates to 1.
	By induction, this is equivalent to saying $w \not \sqsubseteq r_a$ and $w \not \sqsubseteq r_b$ for some successor $w$ of $v$,
	which in turn is equivalent to $v \not \sqsubseteq r$.
	
	Now let $i$ be even, i.e., $v$ is a $\wedge$-gate, which has an $a$-successor $v_1$ and a $b$-successor $v_2$.
	Let $r_a$ and $r_b$ be the roots of $T_{i,a}$ and $T_{i,b}$, respectively. 
	The node $v$ evaluates to 1 in $C$ if and only if both successors $v_1$ and $v_2$ of $v$ in $G$
	evaluate to 1.
	By induction, this is equivalent to saying $v_1 \not \sqsubseteq r'$ and $v_2 \not \sqsubseteq r'$
	where $r'$ is the root of $T_{i-1}$.
	Since the root of $T_*$ simulates every node in an odd level of $G$, this is equivalent to $v \not \sqsubseteq r_a$
	and $v \not \sqsubseteq r_b$.
\end{proof}

\section{Conclusion}

We proved the following results:
\begin{enumerate}
	\item The bisimulation problem for trees that are given by pointer structures (resp., in term representation) 
	is complete for deterministic logspace (resp., $\NC^1$). These results also hold for the simulation problem for trees.
	\item Already for graphs of bounded path-width (a subclass of the graphs of bounded tree-width), 
	the simulation problem becomes $\P$-complete. 
	\item The bisimulation problem between a tree and a dag (or arbitrary graph) belongs to $\AC^1$ and is $\NL$-hard.
	\item Simulation of a tree by a dag as well as simulation of a dag by a tree is \LogCFL-complete.
\end{enumerate}
As an application of (1), we showed that equality of hereditarily finite sets  is $\NC^1$-complete.
For the proof of (1) we introduced the new class of tree-shaped circuits and proved that the circuit evaluation problem for 
tree-shaped circuits of bounded width
is in logspace or $\NC^1$, depending on the representation of the circuit. It would be nice to find further
applications of tree-shaped  circuits.

The main open problem that remains is whether the bisimulation problem for graphs of bounded tree-width is in $\NC$ or 
$\P$-complete. Another open problem is the exact complexity of the bisimulation problem for a tree and a dag. Our bounds
(in $\AC^1$ and $\NL$-hard) are not matching. In many situations (e.g. for various classes of infinite state systems), simulation
is harder than bisimulation. We proved that simulation of a tree by a dag (or vice versa) is $\LogCFL$-complete. Hence, one might
conjecture that also the  bisimulation problem for a tree and a dag belongs to $\LogCFL$.

\bibliographystyle{alpha}
\bibliography{refs}

\end{document}